\date{}
\documentclass[11pt]{article}
\usepackage{amsmath,amsthm,amsfonts,amssymb,amsopn,amscd}     
\usepackage{color}

\DeclareMathOperator{\Tr}{Tr}

\DeclareMathOperator{\Stab}{Stab}

\DeclareMathOperator{\Ind}{Ind}

\def\id{{\rm id}}

\def\eps{\varepsilon}
\def\wh{\widehat}
\def\sig{\sigma}

\def \norm#1{{\vert\!\vert #1\vert\!\vert}}

\def\Tr{\mathrm{Tr}\;}
\def \RR {{\mathbb R}}
\def \CC {{\mathbb C}}
\def \ZZ {{\mathbb Z}}
\def \NN {{\mathbb N}}
\def \eins {{\mathbf 1}}

\def\setminus{\smallsetminus}

\def\A{{\cal A}}
\def\B{{\cal B}}

\def\C{{\cal C}}
\def\D{{\cal D}}

\def\M{{\cal M}}
\def\N{{\cal N}}

\def\I{{\cal I}}
\def\H{{\cal H}}
\def\K{{\cal K}}
\def\W{{\mathcal W}}
\def\S{{\cal S}}

\def\cP{{\cal P}}

\def\ad{{\operatorname{ad}}}
\def\emptyset{\varnothing}

\def\SL2{{{\rm SL}(2,\RR)}}

\def\eins{{\bf 1}}

\def\Mob{{\rm\textsf{M\"ob}}}

\def\S2{S^{1(2)}}
\def\Poi{{\cal P_+^\uparrow}}
\def\tPoi{{\widetilde{\cal P}_+^\uparrow}}
\def\cL{{\cal L}}
\def\wt{\widetilde}
\def\ol{\overline}
\newtheorem{theorem}{Theorem}[section]
\newtheorem{lemma}[theorem]{Lemma}

\newtheorem{corollary}[theorem]{Corollary}

\newtheorem{proposition}[theorem]{Proposition}

\theoremstyle{definition} \newtheorem{definition}[theorem]{Definition}

\theoremstyle{remark}

\def\emptyset{\varnothing}
\def\setminus{\smallsetminus}

\def\Z{{\mathbb Z}}

\def\RR{{\mathbb R}}
\def\CC{{\mathbb C}}
\def\NN{{\mathbb N}}
\def\ZZ{{\mathbb Z}}

\def\sl2{{{\rm SL}(2,\RR)}}
\def\psl2{{{\rm PSL}(2,\RR)}}
\def\u1{{{\rm V}(1)}}
\def\su2{{{\rm SV}(2)}}

\def\so3{{{\rm SO}(3)}}

\def\SO{{\mathrm{SO}}}
\def\ind{{\mathrm{Ind}}}

\begin{document}
\title{\vskip-5mm\bf Split property for free \\massless finite helicity fields}
\author{{\sc Roberto Longo\footnote{Supported in part by the ERC Advanced Grant 669240 QUEST “Quantum Algebraic Structures and
Models”, MIUR FARE R16X5RB55W QUEST-NET, GNAMPA-INdAM .}
\sc,\quad Vincenzo Morinelli$^*$}
\\
Dipartimento di Matematica,
Universit\`a di Roma Tor Vergata,\\
Via della Ricerca Scientifica, 1, I-00133 Roma, Italy\\
E-mail: {\tt longo@mat.uniroma2.it,\quad morinell@mat.uniroma2.it}
\\
\phantom{X}\\
{\sc Francesco Preta}
\\
Courant Institute of Mathematical Sciences,\\ New York University, \\251 Mercer Street, New York, NY 10012, USA\\
E-mail: {\tt preta@cims.nyu.edu}
\\
\phantom{X}\\
{\sc Karl-Henning Rehren$^*$}
\\
Institut f\"ur Theoretische Physik, Universit\"at G\"ottingen,\\ 
37077 G\"ottingen, Germany\\
E-mail: {\tt rehren@theorie.physik.uni-goettingen.de}
}

\maketitle
\begin{abstract}
We prove the split property for any finite helicity free quantum
fields. Finite helicity Poincar\'e representations extend to the
conformal group $\C$ (cf.\ \cite{Mack}) and the conformal covariance
plays an essential role in the argument: the split property is ensured
by the trace class condition $\Tr (e^{-\beta L_0})<+\infty$ for the
conformal Hamiltonian $L_0$  of the M\"obius covariant restriction of
the net on the time axis. We extend the argument for the scalar case presented in \cite{BDL}. We provide the direct sum decomposition into irreducible representations of the conformal extension of any helicity-$h$ representation to the subgroup of transformations fixing the time axis. Our analysis provides new relations among finite helicity representations and suggests a new construction for representations and free quantum fields with non-zero helicity. 
\end{abstract}

\section{Introduction}
\setcounter{equation}{0}
The split property in quantum field theory can be viewed as a strong version of
locality. Locality (= Einstein causality) requires the bounded observables  
localized in two spacelike separated regions $O_1$ and $O_2$ to
generate two commuting von Neumann algebras $\A(O_1)$ and $\A(O_2)$. The split 
property demands that the algebra generated by $\A(O_1)$ and
$\A(O_2)$ is naturally isomorphic to the tensor product $\A(O_1)\otimes
\A(O_2)$ and this can hold only if there is some finite positive distance
between the regions $O_1$ and $O_2$, due to UV
singularities that arise when the regions touch. 

Physically, the split property is motivated as a ``statistical
indepence'' in the sense that states can be independently prepared in
$O_1$ and $O_2$: for every pair of normal states on $\A(O_i)$, there
is a normal state of the full QFT that on $\A(O_i)$ coincides with the given
states \cite{B74}. The relative tensor product position is also an indispensible
prerequisite without which a notion of entanglement of states between
the two subsystems cannot be defined \cite{HS}. 

The terminology ``split'' for a pair of commuting algebras actually
refers to the inclusion of one algebra in the commutant of the other,
asserting that there exists a type $I$ factor\footnote{i.e.,
isomorphic to $\B(\H)$ of some Hilbert space $\H$.} $\B$ such that 
\begin{equation*}
\A_1\subset \B\subset \A_2',
\end{equation*}
cf.\ Definition \ref{def:split}. Because local algebras in QFT are in general
type III (a characteristic feature of QFT as compared to quantum
mechanical systems), the split property is not ensured by basic
assumptions.  

Whether the split property holds for two local algebras at a finite
distance, is a feature of the QFT model under consideration. It has
been verified in various models in quantum field theory, see, e.g.,
\cite{BW86,MTW,BDL}. The split property may fail for topologically non-trivial spacetimes \cite{HO}. Several sufficient conditions are known in terms of
the trace-class property of certain operators related to phase space
\cite{BW86, BDLI, BDL}, indicating that typically, ``too many degrees
of freedom'' may cause it to fail. A deep mathematical understanding
of the split property was given in \cite{lodo}.   

For the massless scalar free field in four spacetime dimensions, 
the split and nuclearity properties for an inclusion of non-touching
double cone regions has been established  in \cite{BDL}. The argument
is essentially group theoretic: the one-particle space of the massless
free field carries an irreducible representation $U$ of the Poincar\'e
group that extends to the conformal group $\C$ in four dimensions. 
$\C$ is the 15-dimensional Lie group generated by the Poincar\'e group
and the ``conformal inversion'' $I$, cf.\ (\ref{I}). It contains the dilations
and the special conformal transformations $I\circ t\circ I$, where $t$ is a 
translation. 

The three-dimensional subgroup generated by time translations and the
conformal inversion is isomorphic to the M\"obius group
$\Mob=\mathrm{SL}(2,\RR)/\ZZ_2$, and acts geometrically on the time
axis $\vec x=0$ exactly like the conformal symmetry group of a chiral
conformal QFT. This means that a conformal quantum field theory in
four dimensions, when restricted to the time axis, becomes a chiral
conformal QFT. In the scalar case, the chiral currents of this theory
are the free scalar field restricted to the time axis, along with all
its spatial derivatives $\nabla_{a_1}\dots\nabla_{a_k}\varphi(t,0)$. 
Their scaling dimensions increase with the number of spatial derivatives. 

The number of quasiprimary (i.e., $\Mob$-covariant) currents as a function of
their scaling dimension is controlled by representation theory. In
this way, the authors of \cite{BDL} could establish that the operator
$e^{-\beta L_0}$ has a finite trace, where $L_0$ is the conformal
Hamiltonian of this chiral conformal QFT.  

This suffices to establish the split property for the algebra inclusions
$\A(O)\subset\A(\wt O)$ when $O\Subset \wt O$ are two
double cones with apices on the time axis. This implies the statistical
independence of $\A(O_1)$ and $\A(O_2)$ whenever $O_1=O$ and $O_2$
is contained in the causal complement of $\wt O$, and then, by
covariance, whenever $O_1$ and $O_2$ are spacelike separated
double cones with a finite distance. 

We adapt this argument to all massless free field theories of finite
helicity, including the free Maxwell field. Because $\Mob$ commutes
with the subgroup $\so3$ of  
spatial rotations, the proof reduces to the computation of the
restriction of the unitary representation of the conformal group on
the one-particle space to the subgroup $\Mob\times \so3$, where the
representations of $\so3$ just provide multiplicities for the
irreducible representations of $\Mob$. The traces of $e^{-\beta L_0}$
in irreducible representations of $\Mob$ are well known, and the
trace-class property on the one-particle space is obtained by an
explicit computation.  This also implies the $L^2$-nuclearity property.

The split property ensures the existence of local unitaries
$U\in\A(O_1)$ that implement inner symmetries on the observables
$a\in\A(O)$ if $\A(O)\subset \A(O_1)$ is split \cite{DL}. Such
operators are usually thought of as (abstract versions of)
$U=e^{iJ^0(f)}$ where $J^\mu$ is an associated conserved local current
and $f$ a suitable test function supported in $O_1$. Indeed such
objects can be rigorously constructed and satisfy the local current
algebra relations \cite{DL}. They thus serve as substitutes for the
covariant massless higher-helicity fields that do not exist by the
Weinberg-Witten theorem \cite{WW}.

Our computation leads to a remarkable observation: as a representation
of $\Mob\times\so3$, the one-particle space for helicity $h+1$ is
just a subrepresentation of that for helicity $h$ (provided
$h>0$). This suggests some new kind of ``deformation argument'' to
construct helicity $h+1$ from helicity $h$, cf.\ Sect.\ \ref{s:defo}.

\section{Preliminaries}
\setcounter{equation}{0}
\subsection{Minkowski spacetime and the Poincar\'e group}
Let $\RR^{1+3}$ be Minkowski space, i.e., $\RR^4$ endowed with the metric $$(x,y)=x_0y_0-\sum_{i=1}^3x_iy_i.$$
In a 4-vector  $x=(x_0,x_1,x_2,x_3)$,   $x_0=t$ and $\vec x=
\{x_i\}_{i=1,2,3}$ are the time and space coordinates,
respectively. The metric induces a causal structure, in particular the
future $x+V_+$ of a point $x$, where $V_+=\{y\in\RR^4:
(y,y)>0,y_0>0\}$. The causal complement of a region $O$ is given by
$O'=\{x\in\RR^{1+3}: (x-y,x-y)<0, \forall y\in O \}$. A {\bf causally
  closed} region is such that $O=O''$. Particularly nice causally
closed regions are the open double cones of the form $O=x_-+V_+\cap
x_+-V_+$, where $x_+$ is a point in the future of $x_-$.

The {\bf Poincar\'e group} $\mathcal{P}$ is the inhomogeneous symmetry
group of $\RR^{1+3}$. It is the semidirect product of the Lorentz
group $\cL$, the homogeneous Minkowski symmetry group, and the
translation group $\RR^4$, i.e., $\mathcal{P}=\cL\ltimes\RR^4$.
We shall indicate with $\Poi=\cL_+^\uparrow\ltimes\RR^4$
the connected component of the identity, 
with $\tPoi$ and  $\widetilde{\mathcal{L}}_+^\uparrow$ the universal coverings resp.\ of $\Poi$ and $\mathcal{L}_+^\uparrow$, and with $\Lambda$ the covering map.

The conformal group $\C$ in four spacetime dimensions is the extension of
the Poincar\'e group by the ``conformal inversion''
\begin{equation}\label{I}
I: (t,\vec x)\mapsto \frac{(-t,\vec x)}{t^2-\vec x^2}.
\end{equation}
Notice that $I$ is singular on $\RR^{1+3}$, but one can extend Minkowski
space to the ``Dirac manifold'' on which $\C$ acts without
singularities, and of which Minkowski space is a dense chart. 
$\C$ is a 15-dimensional Lie group isomorphic to
$\mathrm{SO}(2,4)_0$. The time reversal in the
  numerator of (\ref{I}) ensures that $I$ belongs to the 
connected component.

\subsection{Massless representations of the Poincar\'e group}
The characters of the translation group are $x\mapsto
  \chi_q(x)=e^{i(x,q)}$ where $q\in\RR^{4}$ is a momentum. 
According to Wigner \cite{Wi}, irreducible positive-energy
representations of $\tPoi$ are induced by what is now called Mackey
induction, from irreducible representations of the stabilizer subgroup
(also known as the ``little group'') of some $q$ appearing
in the representation. The characters appearing in
massless positive-energy representations of $\tPoi$ are given by
$q\neq 0$ contained in $\partial V_+=\{x\in\RR^{1+3}: (x,x)=0, x_0\geq 0\}$. 
We fix 
\[
q\equiv(1,0,0,1)\in \partial V_+
\]  
($\partial V_+\setminus\{0\}$ is a $\cL_+^\uparrow$-orbit). We shall
call $\Stab_q$ and $\overline\Stab_q$ the stabilizers
of the point $q$ through the $\widetilde \cL_{+}^\uparrow$ and $\tPoi$
actions, respectively. The latter is the semidirect product of
$\RR^{3+1}$ and the little group $\Stab_q$, i.e., 
$\overline{\Stab}_q=\Stab_q\ltimes\RR^4$.
Any massless $\tPoi$ unitary positive energy representation is
obtained  inducing by a unitary representation of the
$\overline\Stab_q$ group. Note that a $\overline{\Stab_q}$
representation is of the form $\Stab_q\ltimes\RR^4\ni
(x,\sigma)\mapsto V(\sigma)\chi_q(x)$ where $V$ is the unitary representation of  $\Stab_q$.

The little group $\Stab_q $ is isomorphic to $\widetilde E(2)$, the double cover of the Euclidean group of the two-dimensional Euclidean space, namely $E(2)=\mathbb{T}\ltimes\RR^2$. 
Let $U=\Ind_{\overline\Stab_q\uparrow\tPoi}V\, q$ be a unitary representation of $\tPoi$ induced from the representation
$V q$ of $\overline\Stab_q$. In case $V$ is
trivial on the translation subgroup of $\wt E(2)$, $U$ has \emph{finite
  helicity} (or finite spin); in the other cases, it has \emph{infinite spin}.

An irreducible finite helicity representation is of the form
$$U_{h}=\Ind_{\overline\Stab_q\uparrow\tPoi}V_{h}\,\chi_q,\qquad h\in\frac12\,\ZZ$$
where $V_h(g,x)=h(g)$ where ${h}$ is the one-dimensional representation of the double covering of ${\mathbb T}$ of character $2h\in\Z$ ($V_h$ has to be trivial on the translation subgroup of $\wt E(2)$). $h$ is called {\it helicity} parameter.

Massless representations of $\tPoi$ of finite helicity extend to
unitary representations $\widetilde U$ of the conformal group $\C$. The
main argument in our paper pertains to the restriction of this extension to the
M\"obius subgroup of $\C$ (Sect.\ \ref{s:mob}). 
We denote by $P_\mu$ the generators of the
translations, and $K_\mu := \widetilde U(I)P^\mu \widetilde
U(I)$\footnote{Sic! The Lorentz indices are correct due
    to the presence of the time reversal in $I$, cf.\ (\ref{I}).} the generator of the special conformal transformations. Then $i[P_\mu,K_\nu] = -2\eta_{\mu\nu} D +2 M_{\mu\nu}$ where $D$ and $M_{\mu\nu}$ are the
generators of the dilations and Lorentz transformations. 
The conformal Hamiltonian $L_0=\frac12(P_0+K_0)$ generates the
rotations in $\mathrm{SO}(2)\oplus\eins_4\subset \mathrm{SO}(2,4)_0\simeq\C$,
and $\widetilde U(I)=e^{i\pi L_0}$.

\subsection{The M\"obius group and its representations}
\label{s:mob}
{\bf The M\"obius group}. 
The M\"obius group $\Mob$ is the three-dimensional Lie group
$\mathrm{PSU}(1,1)=\mathrm{SU}(1,1)/\ZZ_2$ acting on $S^1\subset \CC$
by fractional linear transformations 
$$S^1\ni z\mapsto\frac{\alpha z+\beta}{\overline\beta z +
  \overline\alpha},\qquad\begin{pmatrix}\alpha & \beta\\
  \overline\beta & \overline\alpha\end{pmatrix}\in \mathrm{SU}(1,1).$$
Via the Cayley transform and its inverse, the stereographic projection:
$$C:\ol\RR = \RR\cup\{\infty\}\ni x\mapsto -\frac{x-i}{x+i}\in S^1,\qquad C^{-1}:S^1\ni
z\mapsto -i\,\frac{z-1}{z+1}\in\ol\RR$$
it is isomorphic to  
$\mathrm{PSL}(2,\RR)=\mathrm{SL}(2,\RR)/\ZZ_2$ acting on the
compactified real line $\ol \RR$ via  
$$\ol\RR\ni
x\mapsto\frac{ax+b}{cx+d}\in\ol\RR,\qquad\begin{pmatrix}a & b\\
  c & d\end{pmatrix}\in \mathrm{SL}(2,\RR).$$ 
We shall freely switch between the ``circle picture'' and the ``line
picture''. 

\medskip

$\Mob$ arises as the subgroup of the conformal group $\C$ in four spacetime
dimension, generated by time translations and the conformal inversion
(\ref{I}). It preserves the time axis and commutes with $\so3$, the
spatial rotations.  
Its more familiar appearance in quantum field theory is in the (unbroken)
conformal group in two spacetime dimension, that is isomorphic to
$\Mob\times\Mob$ acting on the two-dimensional Dirac
manifold $S^1\times S^1$ where each $S^1$ is the compactification of
one lightlike axis.

$\Mob$ can be generated by various one-parameter subgroups. Firstly, consider the following subgroups:
\begin{itemize}\itemsep0mm
 \item Rotations $r: [0,2\pi]\ni\theta\mapsto e^{i\theta}z\in S_1$,  in the circle picture;
 \item Dilations $\delta:\RR\ni s\mapsto e^sx\in \RR$, in the line picture.
 \item Translations $t:\RR\ni s\mapsto x+s\in \RR$, in the line picture.
 \end{itemize}
They are respectively denoted with {\bf K}, {\bf A} and {\bf N}. Any
element $g\in\Mob$ can be uniquely decomposed following the {\bf KAN}
decomposition (Iwasawa decomposition), i.e., let $g\in\Mob$ then
$g=kan$, $k\in{\bf K}$, $a\in{\bf A}$, $n\in{\bf N}$. The subgroup
{\bf A} preserves the upper semicircle, or the right half-line on the
line picture, while ${\bf N}$ maps it into itself for $s>0$. By the
adjoint action of $\Mob$ one can define translation and dilation
groups relative to any other interval $I$, resp.\ $\tau_I$ and
$\delta_I$. 

Another convenient choice replaces the rotations by the special
conformal transformations $I\circ t\circ I$, where the conformal
inversion $I:t\mapsto -1/t$ in the line picture is the rotation by
$\pi$ in the circle picture. 

\medskip

{\bf Unitary positive-energy representations of $\widetilde\Mob$.} Let
$U$ be a unitary representation of $\widetilde\Mob$ on a Hilbert
space $\H$. The self-adjoint infinitesimal generator of the rotation
subgroup in $U$ is denoted by $L_0$, i.e., $U(r(\theta))=e^{i\theta
  L_0}$.  $L_0$ is called the {\it conformal Hamiltonian}. 
Let $P$, $D$, $K$ be the generators of the translations, dilations and
special transformations, resp., and (by abuse of
notation) $I$ the unitary representative of the conformal inversion,
then one has  
\begin{equation}
\label{PDKI}
IPI=K,\quad IDI=-D,\quad L_0=\frac12(P+K),\quad I=e^{i \pi L_0}.
\end{equation}
$U$ is said to be a {\bf positive-energy representation of $\widetilde\Mob$}
if the spectrum of the conformal Hamiltonian $L_0$ is contained in
$[0,+\infty)$.  

Irreducible, unitary, positive energy representations of $\widetilde
\Mob$ on a Hilbert space $\H$ are labelled by positive real numbers
$k$. They correspond to the lowest eigenvalue of the conformal
Hamiltonian $L_0$, called ``lowest weight''.  An irreducible positive unitary representation of
$\widetilde{\Mob}$ factors on $\Mob$, iff $k$ is an integer. 

Let $\mathbf{P}$ be the translation-dilation subgroup of $\Mob$ associated
to $\RR^+$.  A unitary representation of $\mathbf{P}$ is said to have
\textbf{positive energy} if the spectrum of the translation subgroup
is contained in the positive half-line $[0,+\infty)$. There exists a
unique, up to unitary equivalence, irreducible unitary positive-energy representation $U$ of $\mathbf{P}$.
The positivity of the energy of a $\widetilde\Mob$ representation $U$
is equivalent to  the positivity of the translation generator, thus a
positive-energy representation $U$ of $\Mob$ restricts to the unique
positive-energy representation of $\mathbf{P}$ \cite{GLW}. Furthermore, if $U$ is irreducible then $U|_{\mathbf{P}}$ is irreducible \cite{L}.

\subsection{(Anti-)unitary extensions}

{\bf The Poincar\'e group.} Let $\theta$ be the space and time
reflection $(t,\vec x)\mapsto(-t,-\vec x)$ and $\alpha$ be the action of $\theta$ on $\Poi$ by conjugation, we define 
$$\cP_+=\ZZ_2\ltimes_\alpha \Poi$$ to be the extension of  $\Poi$ through $\alpha$. 
 An (anti-)unitary representation of $\cP_+$ is unitary, resp.\  anti-unitary, on $\Poi$ resp.\ on $\theta\,\Poi$.
\begin{proposition}\cite{var}
A unitary irreducible positive energy representation $U$ of $\Poi$ extends (anti-)unitarily to $\cP_+$ iff it is induced by a self-conjugate
representation of the little group.
\end{proposition}
This is true for all irreducible positive energy representations except for those of non-zero finite helicity.
On the other hand $U_h\oplus U_{-h}$ extends to $\mathcal\cP_+$.

\medskip

{\bf The $\Mob_2$ group.} Let $r$ be the complex conjugation
$z\mapsto\ol z$ on $S^1$ ($x\mapsto -x$ in $\ol \RR$), and $\alpha$ be the action of $r$ on $\Mob$ by conjugation, we define 
$$\Mob_2=\ZZ_2\ltimes_\alpha \Mob$$ to be the extension of  $\Mob$ through $\alpha$. 
Note that $r$ reverses the orientation. An (anti-)unitary
representation of $\Mob_2$ is unitary, resp.\ anti-unitary, on $\Mob$
(the orientation preserving transformations of $\Mob_2$) resp.\ on $r\Mob$ (the orientation reversing transformations of $\Mob_2$).
\begin{proposition} \cite{L}
Every unitary positive energy representation $U$ of $\Mob$ extends (anti-)unitarily to $\Mob_2$.
\end{proposition}
Now we are going to show that there exists a unique, up to unitary
equivalence, way to represent (anti-)unitarily such extensions (see
also \cite[Thm.\ 2.11]{NO}).
Let $K$ be a locally compact group, $\alpha$ be an involutive automorphism of $K$ and $G$ be the semidirect product $\ZZ_2\ltimes_\alpha K$. Let $U$ and $\widehat U\dot=U\circ\alpha$ be unitary representations of $K$ on a Hilbert space $\H$  and $J$ be an anti-unitary operator on $\H$, we shall call $J\widehat U J^*$ the \textit{conjugate representation} of $U$. The unitary equivalence class of the conjugate representation does not depend on the choice of $J$. If $\alpha =1$, then our definition of conjugate representation coincides with the classical one. An (anti-)unitary representation of $G$ is unitary on $K$ and anti-unitary on $rK$, where $r$ is the $\ZZ_2$-generator.
$U$ is said to be \textit{self-conjugate} if $U$ is unitarily equivalent to $J\widehat UJ^*$, and \textit{real} if the anti-unitary $J$ can be chosen
s.t.\ $J^2=1$ and  
\begin{equation}\label{eq:real}
U=J\widehat U J^*.
\end{equation} 
Note that such an anti-unitary involution extends the representation
$U$ (anti-)unitarily from $K$ to $G$ (the converse is also true). In
case $J$ can be chosen s.t.\ $J^2=-1$ and \eqref{eq:real} holds, then $U$ is said to be \textit{pseudo-real}. 
\begin{proposition}\label{prop:unicom}
Assume that $K$ is a  locally compact  type I group, and let $U$ be a
unitary representation of $K$. Then
\begin{enumerate}\itemsep0mm
\item If $U$ is real, then it extends to an (anti-)unitary representation of $G$ on $\H$. The extension is unique modulo unitary equivalence. 
\item In general, let $J$ be an anti-unitary involution on $\H$, then $U\oplus J\widehat UJ$ is real and it extends uniquely (up to unitary equivalence) to an (anti-)unitary representation of $G$ on $\H\oplus\H$ as a consequence of point 1.
\end{enumerate}
\end{proposition}
\begin{proof}
Firstly, we  consider the factorial case, namely  $U=U_0\otimes \textbf{1}$ on the Hilbert space $\H=\H_0\otimes\K$, where $U_0$ is an irreducible unitary representation of $K$. We consider the following cases:
\begin{enumerate}\itemsep0mm
\item[(1.a)] Assume that $U$ and $U_0$ are real representations. 
In this case, $U_0$ extends to an (anti-)unitary representation of $G$ through an anti-unitary operator $J_0$ satisfying \eqref{eq:real}. Let $\textbf{J}$ be any antilinear involution on $\K$, one can define the anti-unitary\footnote{The tensor product of two
    operators is defined by $(A\otimes B)(v\otimes w):= Av\otimes
    Bw$. This is ill-defined if $A$ is unitary and $B$ is
    anti-unitary, because it conflicts with $\lambda u\otimes w =
    u\otimes \lambda w$; and it is (anti-)unitary if both $A$ and $B$
    are (anti-)unitary. We thank the referee for asking the question.} involution $J=J_0\otimes \textbf{J}$,  which extends $U$ (anti-)unitarily to $G$ and satisfies \eqref{eq:real}.

Now we have to show the uniqueness up to unitary equivalence of the extension. Consider another anti-unitary involution
$J'$ on $\H$ extending $U$ from $K$ to $G$. The composition $JJ'\in
U(K)'=\CC\otimes \B(\K)$ and  since $JJ'$ is a unitary,  
$$J'=(1\otimes Z)(J_0\otimes \textbf{J})=J_0\otimes Z\textbf{J}$$
where $Z$ is a unitary operator on the Hilbert space $\K$. By uniqueness, up to unitary equivalence, of the complex structure of an Hilbert space, then there exists a unitary  $V\in\mathcal U(\K)$ s.t. $V\textbf{J}V^*=Z\textbf{J}$, thus $$(1\otimes V){J}(1\otimes V^*)={J'}.$$ The two extensions through $J$ and $J'$ are unitarily intertwined by $1\otimes V$.
\item[(1.b)] Assume that $U$ is real and $U_0$ is pseudo-real, w.r.t.\ an anti-unitary operator $J_0$. Let $\textbf{J}$ be an anti-unitary operator s.t.\ $\textbf{J}^2=-1$ then  $J=J_0\otimes \textbf{J}$ is an involution implementing the $\ZZ_2$-generator on $U$ and satisfying \eqref{eq:real}. The argument  of unitary equivalence of the (anti-)unitary extensions is a slight modification of the previous case.
\item [(1.c)] Assume that $U_0$ is disjoint from the conjugate
  representation. Let $J_0$ be an antilinear involution on $\H_0$, we
  define representation 
$$\wt U=\left(U_0\oplus J_0\widehat U_0J_0\right)\otimes\textbf{1}_\K$$ 
acting on
  $\wt\H=(\H_0\oplus \H_0)\otimes \K$.  $\wt U$ is real w.r.t.\
  the following anti-unitary involution $\wt J$. Let  $\sigma$ be the
  flip operator on $\H_0\oplus \H_0$, i.e., $\sigma
  (\xi\oplus\eta)=\eta\oplus\xi$, $\textbf{J}$ be an antilinear
  involution on $\K$, then we define 
$$\wt J=\left((J_0\oplus J_0)\cdot\sigma\right)\otimes \textbf{J}$$  on $\wt \H$. It extends (anti-)unitarily $\wt U$ from $K$ to $G$.  
Let $\wt J'$ be another anti-unitary involution extending $\wt U$ to $G$, then $\wt J\wt J'\in U(K)'=(\CC\oplus \CC)\otimes B(\K)$. Since $\wt J^2=1$, it is easy to see that there exists a unitary $V\in\mathcal U(\K)$  such that $$(\textbf{1}_{\CC^2}\otimes V)\, \wt J \,( \textbf{1}_{\CC^2}\otimes V^*)=\wt J'$$ by uniqueness of the complex structure of the Hilbert space $\K$, and we conclude this case.
\end{enumerate}
We sketch the proof for the general case. Since $K$ is a type I group
the above result generalizes to direct integrals and direct sums of
factorial representations. Indeed, for $U=\int_X U_x d\mu(x)$ where
$\{U_x\}_{x\in X}$ is a family of factorial representations and
$(X,\mu)$ is a standard measure space, the product of any two
anti-unitary involutions $J$ and $J'$ extending $U$ to $G$ belongs to
$U(K)'$. Then one can conclude the proof by applying the factorial case on
integral fibers. 
\end{proof}

Positive-energy unitary representations of $\Poi$ and $\Mob$ satisfy
the assumptions of Proposition \ref{prop:unicom}. In particular,
positive-energy factorial representations of $\Mob$ 
  belong to case (1.a); massive, scalar massless and infinite spin
$\Poi$-representations also belong to (1.a), and massless non-zero
 helicity representations to (1.c).

\section{One-particle nets and Brunetti-Guido-Longo construction}
\setcounter{equation}{0}
In QFT, localization is formulated in terms of {\em local nets}, i.e.,
inclusion preserving maps that associate with open spacetime regions
the corresponding quantum structures (algebras or Hilbert spaces, see
below), and Einstein causality is encoded as a feature of these maps.
We introduce the various nets pertaining to our purpose.

The general idea of the connection between nets of algebras on a
complex Hilbert space containing the vacuum vector $\Omega$ and nets
of real Hilbert subspaces is to define, for every spacetime region
$O$, $H(O):=\overline{\A(O)_{sa}\Omega}$.\footnote{$M_{sa}$ are the self-adjoint elements of a
von Neumann algebra $M$.} One may also take the
intersection with the one-particle space $H_1(O)=H(O)\cap \H_1$. In
the free case, one can recover $H(O)$ and also $\A(O)$ from the real
Hilbert spaces $H_1(O)$ by second quantization, and this can be used as
a construction, once $H_1(O)$ are given. Finally, modular theory
allows to define $H_1(O)$ intrinsically in terms of a positive-energy
representation of $\Poi$. Local fields are not used for the specification
of the local standard subspaces,
  and they can actually be constructed from the latter by
  second quantization.

\subsection{Standard subspaces}

A linear, real, closed subspace $H$ of a complex Hilbert space $\H$ is called {\bf cyclic} if 
$H+iH$ is dense in $\H$, {\bf separating} if $H\cap iH=\{0\}$ and 
{\bf standard} if it is cyclic and separating.

Given a standard subspace $H$  the associated {\bf Tomita operator} $S_H$ is defined to be the closed anti-linear involution with domain $H+iH$, given by: $$S_H:H+iH\ni \xi + i\eta \mapsto \xi - i\eta\in H+iH, \qquad\xi,\eta\in H,$$ on the dense domain $H + iH\subset\H$. The polar decomposition $$S_H = J_H\Delta_H^{1/2}$$ defines the positive self-adjoint {\bf modular operator} $\Delta_H$ and the anti-unitary
{\bf modular conjugation} $J_H$. In particular, $\Delta_H$ is invertible and $$J_H\Delta_H J_H=\Delta_H^{-1}.$$

If $H$ is a  real linear subspace of $\H$, the \emph{symplectic complement} of $H$ is defined by
\[
H' \equiv \{\xi\in\H\ ;\ \Im(\xi,\eta)=0, \forall \eta\in H\} = (iH)^{\bot_\RR}\ ,
\]
where $\bot_\RR$ denotes the orthogonal in $\H$ viewed as a real Hilbert space with respect to the real part of the inner product on $\H$.
$H'$ is a closed, real linear subspace of $\H$. If $H$ is standard, then $H = H''$. 
It is a fact that $H$ is cyclic (resp.\ separating) iff $H'$ is separating (resp.\ cyclic), thus $H$ is standard iff $H'$ is standard, and we have
\[
S_{H'} = S^*_H \ .
\]
Fundamental properties of the modular operator and conjugation are
\begin{equation*}
\Delta_H^{it}H = H, \quad J_H H = H' \ ,\qquad  t\in\RR\ .
\end{equation*}
The one-parameter, strongly continuous group $t\mapsto \Delta_H^{it}$ is the {\bf modular group} of $H$ (cf.\ \cite{RV}).

There is a 1--1 correspondence between Tomita operators and  standard subspaces, namely between:
\begin{itemize}\itemsep0mm
\item Standard subspaces $H\subset\H$,
\item Closed, densely defined  anti-linear involutions $S$ on $\H$,
\item Pairs $(J,\Delta)$ of an anti-unitary involution $J$ and  a positive self-adjoint operator $\Delta$ on $\H$ s.t. \begin{equation}\label{eq:TT}
J\Delta J=\Delta^{-1}. 
\end{equation}
\end{itemize}
Namely, given $(J,\Delta)$ one can recover $S:=J\Delta^{\frac12}$ and
$H$ as the real eigenspace of $S$ with eigenvalue $1$.

\smallskip 

We shall need the following results on standard subspaces. 
\begin{lemma}\label{inc}{\rm \cite{L,LN}.}
Let $H\subset \H$ be a standard subspace, and $K\subset H$ a closed, real linear subspace of $H$. 

If $\Delta_H^{it}K=K$, $\forall t\in\RR$, then $K$ is a standard subspace of $\K\equiv \overline{K+iK}$ and $\Delta_H|_K$ is the modular operator of $K$ on $\K$. 
If moreover $K$ is a cyclic subspace of $\H$, then $H=K$.
\end{lemma}
\begin{lemma}\label{inc2}{\rm \cite{L,LN}.}
Let $H\subset \H$ be a standard subspace, and $U$ a unitary on $\H$ such that $UH=H$. Then $U$ commutes with $\Delta_H$ and $J_H$.
\end{lemma}
The following is the one-particle analogue of Borchers' theorem \cite{Bo2}.
\begin{theorem}{\rm \cite{L,LN}.} \label{thm:Borch}
Let $H\subset\H$ be a standard subspace, and $U$ a one-parameter unitary group on $\H$ with positive generator, such that $U(t)H\subset H$, $t\geq 0$. Then $\Delta_H^{is}U(t)\Delta_H^{-is}= U(e^{-2\pi s}t)$.
\end{theorem}

\subsection{Nets on Minkowski spacetime}
\label{s:nets-mink}

Let $U$ be a unitary positive energy representation of the Poincar\'e
group on a Hilbert space $\H$.  

A $U$-covariant \emph{net of standard subspaces} $\H$ on the set
$\W$ of wedge regions of the Minkowski spacetime is a map
\[
H: \W\ni W \longmapsto H(W)\subset\H
\]
that associates a closed real linear subspace $H(W)$ with each $W\in\W$, satisfying:
\begin{enumerate}\itemsep0mm
\item {\it Isotony}: If $W_1\subset W_2$ then $H(W_1)\subset H(W_2)$;
\item {\it Poincar\'e covariance}:  $U(g)H(W)=H(gW)$ ($W\in\W$, $g\in \Poi$);
\item {\it Reeh-Schlieder property}: $H(W)$ is cyclic $\forall \ W\in\W$;
\item {\it Locality}: 
For every wedge $W\in\W$ we have
 \[
 H(W')\subset H(W)'.
 \]
\end{enumerate}
The net is said to have the Bisognano-Wichmann property if 
\begin{enumerate}\itemsep0mm
\item[5.] {\it Bisognano-Wichmann property}: 
$\Delta^{it}_{H(W)}=U\big(\Lambda_W(-2\pi t)\big)$ for all $W\in\W$
and $t\in\RR$,  where $\Lambda_W$ is the boost subgroup
  fixing the wedge $W$ in the standard parametrization.
\end{enumerate}
Given a $U$-covariant net $H$ on $\W$,
one gets a net of closed, real linear subspaces on double cones $O$
defined by
\begin{equation}\label{HO2}
 H(O)\equiv\bigcap_{\W\ni W\supset O}H(W) \ .
\end{equation}
Note that $H(O)$ is not necessarily cyclic. If $H(O)$ is cyclic and
$H$ has the BW property, then
\[
H(W) = \overline{\sum_{O\subset W} H(O)}
\]
by Lemma \ref{inc}.

\subsection{Nets on the circle}
\label{s:nets-circ}

Let $\I$ be the set of nonempty, nondense, open connected intervals of
the unit circle $S^1= \{z \in \CC : |z|=1\}$.  Let $U$ be a
positive-energy representation of $\Mob$ on a Hilbert space $\H$.

\medskip

A \textbf{M\"obius covariant net} is a map $H$ which assigns to every
interval $I\in\I$ a von Neumann algebra $H(I)\subset \H$  satisfying the following properties:
\begin{enumerate}\itemsep0mm
\item {\it Isotony:} If $I_1,I_2\in\I$ and $I_1\subset I_2$, then  $H(I_1)\subset H (I_2)$;
\item {\it M\"obius covariance:} 
$U(g)H(I)=H(gI)$ ($I\in\I$, $g\in\Mob$);
\item 
{\it Reeh-Schlieder property:} $H(I)$ is cyclic for every $I\in\I$;
\item {\it Locality:} If $I_1,I_2\in\I$ and $I_1\cap I_2=\emptyset$, 
\[
H(I_1)\subset H(I_2)'.
 \]
\end{enumerate}
With these properties, $H(I)$ are standard subspaces,
and the net automatically satisfies the Bisognano-Wichmann property 
\begin{enumerate}\itemsep0mm
\item[5.] {\it Bisognano-Wichmann property}: 
$\Delta^{it}_{H(I)}=U\big(\delta_I(-2\pi t)\big)$ for all $
I\in\I$ and $t\in\RR$.
\end{enumerate}
Here, for $I_+$ the upper semicircle, $\delta_{I_+}(t)$ are the
dilations $x\mapsto e^tx$ in the line picture, and for every other
interval $I=g(I_+)$ ($g\in\Mob$), $\delta_I(t)=
g\circ\delta_{I_+}(t)\circ g^{-1}$. 

\medskip

A translation-dilation covariant net of standard subspaces on the
intervals of the real line $\RR$ can be defined in complete analogy. 
It is said to satisfy the Bisognano-Wichmann property if
$U(\delta_{\RR_+}(2\pi t))=\Delta_{\RR_+}^{-it}$. In this case, it is
possible to obtain a net on the circle; and also the converse is true: 
\begin{lemma}\cite{L} Let $H$ be a translation-dilation net on the
  line. It extends to a M\"obius covariant net on the circle if and
  only if the Bisognano-Wichmann property holds. The extension is unique.
\end{lemma}

\subsection{Brunetti-Guido-Longo construction}
\label{s:BGL}
The Brunetti-Guido-Longo construction relies on the 1--1 correspondence between standard subspaces and Tomita-Takesaki modular data.

\medskip

 \textbf{On Minkowski space.} \cite{BGL} Let $U$ be an
(anti-)unitary representation of $\mathcal{P}_+$, and $J_W$ and $K_W$
the anti-unitary reflection and the self-adjoint generator of the
one-parameter group of boosts associated with the wedge $W$
(i.e., $U(\Lambda_W(t))=e^{it K_W}$), respectively. The pair
$(J_W,\Delta_W\equiv e^{-2\pi K_W})$, satisfies \eqref{eq:TT}, thus one can
associate to any wedge $W\in\W$ a standard subspace $H(W)$ in a
covariant way. By positivity of the energy and the Borchers Theorem
\ref{thm:Borch}, the net $W\mapsto H(W)$
satisfies Isotony. Covariance, Locality and the Reeh-Schlieder and Bisognano-Wichmann
properties hold by construction. 

\medskip

\textbf{On the circle.} \cite{L} Let $U$ be an (anti-)unitary
representation of $\Mob$ and $J_I$ and $K_I$ the anti-unitary
reflection and the generator of the one-parameter group of dilations
associated with the interval $I$, respectively, and 
$\Delta_I=e^{-2\pi K_I}$. In analogy with the previous, we
can define a net $I\mapsto H(I)$. By positivity of the energy and
the Borchers theorem, the net $I\mapsto H(I)$ satisfies Isotony. Covariance, Locality and the Reeh-Schlieder and Bisognano-Wichmann
properties hold by construction. 
\begin{proposition}\label{prop:uuu}
There is a unique, up to unitary equivalence, net of standard
subspaces on the considered spacetimes (circle or Minkowski)
satisfying 1.--5. of Sect.\ \ref{s:nets-mink} resp.\ \ref{s:nets-circ}.  \end{proposition}
\begin{proof}
{\it On Minkowski space.} Let $W\mapsto H(W)$ be a Poincar\'e
covariant net of standard subspaces on wedges  satisfying the
Bisognano-Wichmann property w.r.t.\ a positive-energy unitary
Poincar\'e representation $U$. Then the modular conjugations of wedge
subspaces extend $U$ to an (anti-)unitary representation of $\mathcal
\cP_+$, cf.\ \cite{GL95}. We conclude by the unitary equivalence of
(anti-)unitary extensions in Proposition \ref{prop:unicom} and the 1--1 correspondence between Tomita operator and standard subspaces.

\smallskip

{\it On the circle.} Let $I\mapsto H(I)$ be a M\"obius covariant net
of standard subspaces on intervals satisfying 1.--5. Then $U$ extends
to an (anti-)unitary representation of $\Mob_2$ through interval
modular conjugations \cite{L}. The conclusion again follows by
Proposition \ref{prop:unicom}  and the 1--1 correspondence between Tomita operators and standard subspaces.
\end{proof}

\subsection{Chiral current models \cite{GLW}}
\label{s:U1}

For $n\in\NN$ consider the Hilbert space $\H_n$ defined by the
closure of the space of square integrable functions on 
  $\RR$ w.r.t.\ the inner product   
$$(f,g)_n=\int_0^\infty p^{2n-1}dp\,\ol{\wh f(p)}\wh g(p).$$
(Via the Cayley transform, it can be identified with a space of
square-integrable functions on $S^1$.)
Its null space contains the polynomials of degree $2(n-1)$. The associated
symplectic form on the real-valued functions is
$$\omega_n(f,g)\equiv\Im (f,g)_n=\frac{(-1)^{n-1}}{2}\int 
f(x)g(y)\delta^{(2n-1)}(x-y)\, dx\,dy.$$
$\H_n$ carries a unitary positive-energy representation 
  $U^{(n)}$ of $\Mob$ by 
\begin{equation}\label{mbact}
(U^{(n)}(g)f)(x) = \big(\frac{dg(x)}{dx}\big)^{-2(n-1)}\, f(g(x))=
(cx-a)^{2(n-1)}\, f(g(x)),
\end{equation}
in particular 
\begin{equation}\label{Iact}
(If)(x) = x^{2(n-1)}\cdot f(I(x)).
\end{equation}
The self-adjoint generators act by 
\begin{equation}\label{PD}(Pf)(x)=i\partial_xf(x), \quad
  (Df)(x)=i(x\partial_x-(n-1))f(x),
\end{equation}
\begin{equation}\label{K} (Kf)(x)=i(x^2\partial_x-2(n-1)x)f(x).
\end{equation}
$U^{(n)}$ is the positive-energy representation of lowest weight $n$.

\medskip

Applying the BGL construction (Sect.\ \ref{s:BGL}) to the
representation $U^{(n)}$, one obtains the net of real subspaces
$$I\mapsto H_n(I)=\overline {\left\{f\in \C^\infty
     (\RR,\RR):\mathrm{supp\,}f\subset I\right\}}^{\|\cdot\|_n}\subset\H_n,$$
on which $\Mob$ acts covariantly. (``Modular localization'' is the
fact that the support property arises as a consequence of the
definition of $H_n(I)$ via modular theory. Locality is then seen directly
from the symplectic form $\omega_n$.) 
The second quantization of $H_n$ (cf.\ Sect.\ \ref{s:secq}) gives the
net of von Neumann algebras generated by the quasiprimary chiral
current\footnote{A quasiprimary chiral current of
    dimension $n$ is a field
  on $S^1$ transforming under M\"obius transformations like $U(g)j(z)U(g)^* = (dg(z)/dz)^n\cdot
  j(g(z))$.} $j_n$ of dimension $n$.

\medskip

Note that 
$$(f,g)_n=(\partial_x f,\partial_x g)_{n-1},$$
i.e., the derivative $\partial_x$ is a unitary operator $\H_n\to\H_{n-1}$. This
operator intertwines the actions of the generators $P$ and $D$, but
not of $K$. Thus, $\partial_x:\H_n\to\H_{n-1}$ implements the unitary
equivalence of the restrictions of $U^{(n)}$ to the translation-dilation
subgroup, cf.\ Sect.\ \ref{s:mob}. In the language
of quantum field theory, this is the statement that a quasiprimary
current of dimension $n$ and the derivative of a quasiprimary
current of dimension $n-1$ transform in the same way under
translations and dilations. 

The distinction is only seen by the action
of $K$, or $I$. This motivates the following   
\begin{lemma} \label{l:geo}
Let $U$ be a representation $U$ of $\Mob$ whose 
restriction to the translation-dilation subgroup $\mathbf{P}$ is given by (\ref{PD}). Suppose that $I$ acts
geometrically, i.e., 
\begin{equation}\label{Igeo}
(If)(x)= g(x)f(I(x))
\end{equation}
with some function $g$. Then $g(x)=x^{2(n-1)}$, and $U=U^{(n)}$.
\end{lemma}
\begin{proof} By direct computations, using (\ref{PDKI}): Insertion of (\ref{Igeo}) into
  $ID+DI=0$ implies that $g$ is homogeneous of degree $2(n-1)$, hence
  $g(x)=g_0x^{2(n-1)}$. $I^2=\id$ implies $g_0=\pm1$. Then $K=IPI$
  implies (\ref{K}), which together with $(\ref{PD})$ integrates to
  (\ref{mbact}). Then (\ref{Iact}) implies $g_0=1$. 
\end{proof}

The following is a reformulation of Lemma \ref{l:geo}. 
\begin{proposition}\label{prop:quasi} For $n\in \NN$
    let $\H$ with inner product
  $(f,f)=\int_0^\infty dp \, \vert\widehat f(p)\vert^2 p^{2n-1}$ be the
   anti-Fourier transform of the Hilbert space $L^2(\RR_+,p^{2n-1}dp)$. 
Let 
$$I\mapsto H(I)=\overline{\{f\in\C_0^\infty(\RR,\RR), \mathrm{supp
  }\,f\subset I \}}\subset \H$$ 
be a $\Mob$-covariant net of standard subspaces with the natural action of
translations and dilations on $\H$. Then $H(I)=\overline{\{ j_n(f)\Omega\in\H: \mathrm{supp }\,f\subset I \}}$ where $j_n$ is the quasiprimary field of dimension $n$.  In particular, $H$ is the one-particle net $H_n$ associated with $U^{(n)}$ (up to multiplicity). 
\end{proposition}
\begin{proof}
By the Bisognano-Wichmann property we know that $H$ is
the canonical BGL net associated with the covariant
$\Mob$-representation. Thus there exists a current $j$ generating $H$,
and it remains to identify $j$ with the quasiprimary
current $j_n$ of dimension $n$. 

Suppose that the inner product $(f,f)=\norm{j(f)\Omega}^2$ were
misidentified, say, for simplicity, as
$(f,f)=\norm{j_{n-1}'(f)\Omega}^2$. This is possible since $j_n$ and
$j'_{n-1}=\partial j_{n-1}$ share the same scaling dimension, and the
same translation-dilation covariant representation (but inequivalent
$\Mob$ covariant representations). Then the conformal inversion would act
geometrically on derivatives as $f'$ because $j_{n-1}'(f)=-j_{n-1}(f')$, but not
on its primitive $f$. This is a contradiction.
As a consequence $j=j_n$ and $H=H_n$.
\end{proof}

\subsection{Second quantization and nets of von Neumann algebras}
\label{s:secq} 

With $\H$ a Hilbert space and $H\subset\H$ a real linear subspace,
$R_+(H)$ is the von Neumann algebra on the symmetric Fock space
$\mathrm{F}_+(\H)$ generated by the CCR operators:
\begin{equation}\label{Rpm}
R_+(H) \equiv \{\mathrm{w}(f): f\in H\}'',
\end{equation}
with $\mathrm{w}(f)$ the Weyl unitaries on $\mathrm{F}_+(\H)$ defined 
on the coherent states $e^g\in \mathrm{F}_+(\H)$ ($f\in \H$) by their 
action $\mathrm{w}(f)e^g = e^{-\frac12(f,f)-(f,g)}\cdot e^{f+g}$. If
$\varphi(f)$ is the selfadjoint generator of the unitary one-parameter
group $\mathrm{w}(f)$, this standard construction ensures the identification
of the ``one-particle vector'' $\varphi(f)\Omega \in \mathrm{F}_+(\H)$
with $f\in\H\subset \mathrm{F}_+(\H)$. 
By continuity we have that
\[
R_+(H) = R_+(\ol H)\ .
\]
Moreover the Fock vacuum vector $\Omega$ is cyclic (resp.\ separating)
for $R_+(H)$ iff $\ol H$ is cyclic (resp.\ separating). 

\newpage

Second quantization respects
the lattice structure \cite{A} and the modular structure
\cite{LRT,LMR}. We recall these basic properties. 
For a standard subspace $H\subset\H$, we denote by $S^+_H$, $J^+_H$, $\Delta^+_H$ the Tomita operators associated with $(R_+(H),\Omega)$, and by
$\mathit{\Gamma}_+(T)$ the Bose second quantization of a one-particle
operator $T$ on $\H$, $\mathit{\Gamma}_+(T) e^f = e^{Tf}$.
\begin{proposition}\label{prop:secquant} \cite{A,LRT,LMR}
Let $H$ and $H_a$ be closed, real linear subspaces of $\H$. We have
\begin{itemize}\itemsep0mm
\item[$(a)$] $R_+(H)' = R_+(H')$; 
\item[$(b)$] $R_+(\sum_a H_a) = \bigvee_a R_+(H_a)$;
\item[$(c)$] $R_+(\bigcap_a H_a) = \bigcap_a R_+(H_a)$.
\item[$(d)$] If $H$ is standard, then $S^+_H = \mathit{\Gamma}_+(S_H)$, \ $J^+_H = \mathit{\Gamma}_+(J_H)$, \ $\Delta^+_H= \mathit{\Gamma}_+(\Delta_H)$. 
\end{itemize}
\end{proposition}
Given the canonical BGL-net $H_U$ associated with a unitary
positive-energy representation $U$ of $\Poi$ or of
  $\Mob$, respectively, its second quantization net 
\[
\A(W) \equiv R_+\big(H_U(W)\big)\ ,\quad W\in\W\ , \quad\hbox{resp.}\quad 
\A_n(I) \equiv R_+\big(H_{U^{(n)}}(I)\big)\ ,\quad I\in\I\ , 
\]
is the free field net, i.e., $\A(W)$ is generated by Weyl operators
$\mathrm{w}(f)=e^{i\varphi(f)}$ of free Wightman fields smeared with
real test functions supported in $W$, and
$\A_n(I)$ is generated by Weyl operators
$\mathrm{w}(f)=e^{ij_n(f)}$ of the quasiprimary current
of dimension $n$, smeared in $I$. The case $n=1$ is the canonical
$U(1)$ current. 
 
These nets satisfy the
usual assumptions on nets of von Neumann algebras of local
observables. 

\medskip 

\noindent {\bf On Minkowski space.}
\begin{itemize}\itemsep0mm
\item {\it Isotony}: $\A(W_1)\subset \A(W_2)$ if $W_1\subset W_2$;
\item {\it Poincar\'e covariance}: $U$ is a positive-energy
  representation of $\Poi$, and $U(g)\A(W)U(g)^* = \A(gW)$,\ \  $g\in \Poi$;
\item {\it Vacuum with Reeh-Schlieder property}: there exists a unique
  (up to a phase) $U$-invariant vector $\Omega\in\H$, and $\Omega$ is  
cyclic and separating for $\A(W)$ for all $W\in\W$;
\item {\it Locality}: $\A(W')\subset \A(W)'$.
\end{itemize}
In addition, for the canonical free field nets the {\it Bisognano-Wichmann property} holds:
\[
\Delta^{it}_{\A(W),\Omega}=U\big(\Lambda_W(-2\pi t)\big),\quad 
W\in\W,\; t\in\RR \ , 
\]
where $\Delta_{\A(W),\Omega}$ is the modular operator of $(\A(W),\Omega)$.

\medskip

\noindent {\bf On the circle.}
\begin{itemize}\itemsep0mm
\item {\it Isotony}:  $\A(I_1)\subset \A(I_2)$ if $I_1\subset I_2$;
\item {\it M\"obius covariance}: $U$ is a positive-energy
  representation of $\Mob$, and $U(g)\A(I)U(g)^* = \A(gI)$,\ \  $g\in \Mob$;
\item {\it Vacuum}: There exists a unique
  (up to a phase) $U$-invariant vector $\Omega\in\H$;
\item {\it Locality}: $\A(I')\subset \A(I)'$,  $I\in\I$;
\end{itemize}
The following are consequences of these axioms
\begin{itemize}\itemsep0mm
\item {\it Reeh-Schlieder property}: $\Omega$ is a cyclic and separating vector for each $\A(I)$, $I\in\I$; 
\item{\it Haag duality}: $\A(I')'=\A(I)$,  $I\in\I$; 
\item{\it Bisognano-Wichmann property}: $U(\delta_I(-2\pi t))=\Delta^{it}_{\A(I),\Omega}$,  for all $
I\in\I$ and $t\in\RR$. 
\end{itemize}

\section{Time-axis theory of finite helicity representations}
\setcounter{equation}{0}
Consider the representation $U=U_h\oplus U_{-h}$ of the Poincar\'e group. The
Brunetti-Guido-Longo construction associates with $U$ a net of standard
subspaces $H$ on wedge shaped regions satisfying the Bisognano-Wichmann
property. The second quantization procedure provides the free field
net $\A$ associated with $U$. 

Finite helicity von Neumann algebra nets have an associated Wightman field $\phi_h$ satisfying the Bisognano-Wichmann property \cite{BW}. Thus the BGL and the Wightman field constructions coincide as
$$H(O)=\overline{\{\phi_h(f)\Omega:f\in\C^\infty_0(\RR^{1+3}), \,\mathrm{Supp} f\subset O\}}\quad \text{and}\quad H(W)=\bigcup_{O\subset W}H(O),$$
gives a one-particle $U$-covariant net (with two polarizations $h$ and $-h$)
and its second quantization
$$\A_h(O)\dot=R_+(H(O))=\{e^{i\overline{\phi_h(f)}}:f\in\C^\infty_0(\RR^{1+3}), \,\mathrm{Supp} f\subset O\}\}''$$
gives the free field. Note that Haag duality holds by \cite{HL,H}, namely $H(O')=H(O)'$ and $R_+(H(O'))=R_+(H(O)')=R_+(H(O))'$. Furthermore, due to the conformal covariance, the modular operator of any double cone subspace (resp.\ second quantization algebra) implement a one-parameter group of conformal transformation that is conjugated to the dilation and the boost one parameter groups \cite{H}.

Firstly, note that it is not possible to unitarily rewrite the net $H$
as a direct sum according to $U_h\oplus U_{-h}$, as $U_h$ does not extend (anti-)unitarily to $\cP_+$ \cite{Mo,var}. 
On the other hand $ U_{\pm h}$ (thus $U$) extends to a representation $\widetilde U_{\pm h}$ (resp.\ $\widetilde U$) of the conformal group which acts covariantly on the net $H$, see e.g. \cite{HL, H, Mack}. 
 
We recall that a local net of standard subspaces on double cones undergoing the action of a massless Poincar\'e representation is timelike local.
\begin{lemma}\label{lem:timecom} \cite{LMR}
Assume that $U$ is a massless, unitary representation of $\tPoi$ acting covariantly on a local net of closed, real linear subspaces on double cones. 
Let $O_1, O_2$ be double cones with $O_2$ in the time-like complement of $O_1$, then 
\[
H(O_2)\subset H(O_1)'\ ,
\]
where $H(O)=\bigcap_{W\supset O}H(W)$.
\end{lemma}
Now, we can define a local net of standard subspaces on the time axis. Let $I=(a,b)\subset\RR$ be an interval and $O_I=(V_-+b)\cap (V_++a)$ the double cone with vertices on the time axis. Then we get a net on the line 
$$I\mapsto H(I)=H(O_I)$$
which undergoes the M\"obius covariant action of $\widetilde U|_{\Mob}$.

Since any unitary positive energy  M\"obius representation extends
(anti-)unitarily  to $\Mob_2$, then $\widetilde U_{\pm h}|_{\Mob}$
extends to  $\Mob_2$ and acts covariantly on its BGL net $H_{0}$ of standard subspaces.
By Proposition \ref{prop:uuu} and 
the Bisognano-Wichmann property for the dilation group, we have that the net  $I\mapsto H(I)$ is unitarily equivalent to the direct sum of the two local $\Mob$-covariant nets $H_{\pm}$.

Now we need the structure coming from Wightman fields in order to
construct the theories on the time axis.

\subsection{One-particle space and free field equations}

Free field theories are completely determined by their one-particle
structure. This structure is conveniently described by the
two-point functions of Wightman fields, that define the one-particle space by 
endowing the space of test functions with an inner product. 
The null space of this inner product is completely characterized by the free field equations (that are closer to the physicists' mind). Our strategy is to use the latter in order to
control the one-particle space and the pertinent decomposition of the 
one-particle representations.

As compared to the scalar field, there are two complications with
helicity $>1$: the (higher) Maxwell fields transform nontrivially
under $\so3$, and the one-particle space of a local field carries
necessarily the direct sum of the irreducible representations of
helicity $+h$ and $-h$. (Nevertheless, we shall loosely refer to the
local fields as ``helicity-$h$ fields''.)   

\medskip

\textbf{The electromagnetic field, $h=1$}. 
The Maxwell equations for the magnetic and electric fields in absence of charges are
$$\mathrm{curl}\,\textbf{B}=\frac{\partial }{\partial
  t}\textbf{E},\qquad\mathrm{curl}\,\textbf{E}=-\frac{\partial
}{\partial t}\textbf{B},\qquad \mathrm{div}\,\textbf{E}=0=\mathrm{div}\,\textbf{B}.$$
The field strength $F_{\mu\nu}$ of the electromagnetic field is defined to be the anti-symmetric tensor given by
$$\textbf{E}=(F_{01},F_{02},F_{03})\quad
\textrm{and}\quad\textbf{B}=(F_{32},F_{13},F_{21}),$$ 
and the Maxwell equations become 
$$\partial^{\mu}F_{\mu\nu}=0\quad \text{and}\quad \partial_{\mu}F_{\nu\rho}+\partial_{\nu}F_{\rho\mu}+\partial_{\rho}F_{\mu\nu}=0.$$
These imply the Klein-Gordon equation
$\Box F=0$.

In the quantized theory, the one-particle Hilbert space is the space of test
functions $f^{\mu\nu}$ equipped with the inner product given by the
two-point function
$$(f,f) := (F(f)\Omega,F(f)\Omega)= (\Omega, F(\overline
f)F(f)\Omega).$$
The latter is dictated by covariance (i.e., by Weinberg's quantization \cite{We}
based on Wigner's intrinsic construction
\cite{Wi} avoiding the use of a potential) to be 
$$\|F(f)\Omega\|^2=\int_{\RR^3}
\frac{d\textbf{p}}{|\textbf{p}|}\,p_\mu
\,p_\tau\,\eta_{\nu\sigma}\,\overline{\widehat f^{\mu\nu}}(p)\,\widehat
f^{\sigma\tau}(p).$$
Higher correlations are obtained by Wick's theorem, so that the full
Hilbert space is the Fock space, and multi-particle states can be created by the usual creation and annilation operators.
The field strength transforms covariantly under the Poincar\'e group:
$$U(a,\Lambda)F_{\mu\nu}(x)U(a,\Lambda)^*=\Lambda^\rho_\mu\,\Lambda^\sigma_\nu F_{\rho\sigma}(\Lambda x+a).$$
It is well known that $U(a,\Lambda)$ acts on the one-particle space as the direct
sum of Poincar\'e representations of helicities $1$ and $-1$ \cite{We}.

In order to prove the split property for the resulting net, we want to
restrict the Maxwell net to the time axis. This gives a chiral conformal
QFT. By computing $\Tr e^{-\beta L_0}$ for this chiral QFT and showing
that it is finite for all $\beta>0$, we shall establish that the
chiral net satisfies the split property. From this, we can conclude
that the original net has the split property. 

Before we present the purely representation-theoretical argument for arbitrary
helicities $\vert h\vert\geq 1$, we want to give its field-theoretic version in
the Maxwell case.

The Poincar\'e transformations of the Maxwell tensor extend to the
conformal group by 
$$U(g)F_{\mu\nu}(x)U(g)^*= J_g(x)_\mu^{\rho}J_g(x)_\nu^{\sig}F_{\rho\sig}(g(x)),$$
where $J_g(x)^{\rho}_\mu=\partial g(x)^{\rho}/\partial x^\mu$ is the
Jacobi matrix. For infinitesimal transformations with
generators $P_0$ (time translations), $D$ (dilations) and $K_0 = IP_0
I$ (special conformal transformations), one finds
$$i[P_0,F_{\mu\nu}(x)]=\partial_0F_{\mu\nu}(x),\quad
i[D,F_{\mu\nu}(x)]=(x^\kappa\partial_\kappa + 2)F_{\mu\nu}(x),$$
$$i[K_0,F_{\mu\nu}(x)] = \big(2x_0(x\partial)
-x^2\partial_0+ 4\, x_0\big) F_{\mu\nu}(x) +
2(\eta_{0\mu}x^\kappa F_{\kappa\nu}-x_\mu F_{0\nu} - (\mu\leftrightarrow\nu)).$$
From this, the commutators with the restricted fields
$\nabla_{\underline a}
F_{\mu\nu}(t)=\nabla_{a_1}\dots\nabla_{a_k}F_{\mu\nu}(t,\vec x)\vert_{\vec x=0}$
can be explicitly worked out. 

Now, $P_0=P$, $D$ and $K_0=K$ are the
generators of $\Mob$, and quasi-primary chiral currents of dimension $h$
transform as 
$$i[P,j(t)]=\partial_tj(t),\quad i[D,j(t)]=(t\partial t+h)j(t),\quad
i[K,j(t)]=(t^2\partial t+2ht)j(t).$$
It is obvious that the first two equations are satisfied by 
$\nabla_{\underline a} F_{\mu\nu}(t)$ with $h=2+\vert{\underline a}\vert = 2+k$; but
the last one is in general not fulfilled. In an $\so3$-covariant
formulation, and using the Maxwell equations, we can bring the
commutator with $K$ into the form 
$$i[K,J_{{\underline a},b}(t)] - (t^2\partial_t+2(2+k)t) J_{{\underline a},b}(t) = 2
\sum_{1\leq i<j\leq k}\partial_t\delta_{a_ia_j}J_{\underline {\ddot
    a},b}(t)+2i\sum_{1\leq i\leq k}\eps_{a_ibc}J_{\underline {\dot a},c}(t),$$
where $J_{{\underline a},b}$ are the complex fields
$\nabla_{a_1}\dots\nabla_{a_k}(E_b(t,\vec x)+iB_b(t,\vec
x))\vert_{\vec x=0}$, and ${\underline {\dot a}}$ is the multi-index with $a_i$ deleted, and
similarly $\underline{\ddot a}$ is the multi-index with $a_i$ and $a_j$ deleted.

The quasi-primary currents are those for which the right-hand side
vanishes. It is easy to see that this is precisely the case for the
completely symmetric and traceless part of the rank $k+1$ tensor 
$J_{a_1\dots a_k,b}$. This tensor carries the spin
$s=k+1$-representation of $\so3$, and because $J$ is complex, there
are two $2s+1$ multiplets of real quasi-primary currents of dimension
(= lowest weight of $L_0$) $2+k$. All other components of $J_{a_1\dots a_k,b}$ can be seen to be time
derivatives of lower currents by virtue of the Maxwell equations $\partial_aJ_a=0$, 
$\partial_aJ_b-\partial_bJ_a=i\eps_{abc}\partial_tJ_c$ and the wave
equation that follows from them.

Now, it is well known that on the subspace generated from the vacuum by a
quasi-primary field of dimension $h$, one has $\mathrm{Tr}_h\,e^{-\beta
  L_0}=\frac{e^{-\beta h}}{1-e^{-\beta}}$, hence on the one-partical
space of the Maxwell field, 
$$\Tr e^{-\beta
  L_0}=2\sum_{k\geq0}(2k+3)\cdot\frac{e^{-\beta
    (2+k)}}{1-e^{-\beta}},$$
which can be easily summed as a geometric series in $z=e^{-\beta}$
with radius of convergence $1$.

\bigskip

\textbf{Higher helicity fields, $h>1$}. 
The field strength is a tensor 
$$F_{[\mu_1\nu_1]\dots[\mu_h\nu_h]},$$
anti-symmetric in each index pair $[\mu\nu]$. It transforms
covariantly under the Poincar\'e group: 
$$U(a,\Lambda)F_{[\mu_1\nu_1]\dots[\mu_h\nu_h]}(x)U(a,\Lambda)^*=\Lambda^{\rho_1}_{\mu_1}\,\ldots\,\Lambda^{\rho_h}_{\mu_h}\, \Lambda^{\sigma_1}_{\nu_1}\,\ldots\,\Lambda^{\sigma_h}_{\nu_h} F_{[\rho_1\sigma_1]\dots[\rho_h\sigma_h]}(\Lambda x+a)$$ and is
subject to the linear dependencies (symmetries)
\begin{eqnarray}\label{HMT}F_{\dots[\mu_j\nu_j]\dots[\mu_k\nu_k]\dots} =
F_{\dots[\mu_k\nu_k]\dots[\mu_j\nu_j]\dots},\quad \eta^{\mu_j\mu_k}
  F_{[\mu_1\nu_1]\dots[\mu_h\nu_h]}=0,\notag \\
F_{[\alpha\beta][\gamma\nu_2]\dots[\mu_h\nu_h]} +
  F_{[\beta\gamma][\alpha\nu_2]\dots[\mu_h\nu_h]}+F_{[\gamma\alpha][\beta\nu_2]\dots[\mu_h\nu_h]}=0.
\end{eqnarray}
Its equations of motion (``higher Maxwell equations'') are
\begin{eqnarray}\label{HME}
\partial^\alpha
  F_{[\alpha\nu_1]\dots[\mu_h\nu_h]}=0,\quad\partial_\alpha
  F_{[\beta\gamma]\dots[\mu_h\nu_h]}+ \partial_\beta
  F_{[\gamma\alpha]\dots[\mu_h\nu_h]}+\partial_\gamma
  F_{[\alpha\beta]\dots[\mu_h\nu_h]}=0.
\end{eqnarray}
One can solve the linear dependencies in an $SO(3)$-covariant way by
introducing the ``electric'' and ``magnetic'' components
$$E_{b_1\dots b_h}:=F_{[0b_1]\dots [0b_h]},\quad B_{b_1\dots b_h}:=\eps_{b_1j_1k_1}F_{[j_1k_1][0b_2]\dots [0b_h]},$$
so that both $E$ and $B$ are symmetric
traceless tensors, hence they carry the representation $D^s$ of
$SO(3)$; furthermore, the identities  
$$\eps_{b_1j_1k_1}\eps_{b_2j_2k_2}F_{[j_1k_1][j_2k_2][\mu_3\nu_3]\dots
  [\mu_h\nu_h]} = -F_{[0b_1][0b_2][\mu_3\nu_3]\dots
  [\mu_h\nu_h]}$$
shows that two ``magnetic'' indices amount to two ``electric'' indices
up to a sign, so that the $\so3$ tensors $E$ and $B$ contain all
independent components of the higher Maxwell tensor.

Thus, a general field operator is of the form $F(f)=E(f^E)+B(f^B)$,
where the test function is a pair 
$$f(x)=(f^E_{b_1\dots b_h}(x),f^B_{b_1\dots b_h}(x))$$
of completely symmetric traceless tensors. 

Also the higher Maxwell equations look
the same as for $h=1$, namely $E$ and $B$ are divergence-free and
\begin{eqnarray}\label{HEB}
\eps_{abc} \nabla_{a} E_{bb_2\dots b_h} = -\partial_t
B_{cb_2\dots b_h}, \quad \eps_{abc} \nabla_{a} B_{bb_2\dots b_h}
= \partial_t E_{cb_2\dots b_h}
\end{eqnarray}
(which of course holds in every index).

Test functions that arise by smearing the Maxwell
equations belong to the kernel of the two-point function, and
hence are zero as elements of the one-particle Hilbert space. Thus, in
the one-particle space, there hold linear relations
  among test functions, of the form
\begin{eqnarray}\label{null}
(\nabla_bg_{\underline{\dot b}},0)\doteq0,&\qquad
  (0,\nabla_bg_{\underline{\dot b}})\doteq0,\notag \\(\eps_{abc}\nabla_ag_{c\underline{\dot b}},0)\doteq(0,-\partial_t
g_{\underline b}),&\qquad (0,\eps_{abc}\nabla_ag_{c\underline{\dot b}})
\doteq(\partial_t g_{\underline b},0).
\end{eqnarray}
Because the higher Maxwell equations imply the wave equation, also
$$((\vec\nabla^2-\partial_t^2)g^E,(\vec\nabla^2-\partial_t^2)g^B)\doteq0$$
are zero in the one-particle space.

\subsection{Counting currents}

The space of ``test functions'' for the fields restricted to the time
axis is spanned by $f=(f^E_{\underline b},f^B_{\underline b})$ where\footnote{That Wightman
  fields can be restricted to $\vec x=0$ is a result due to Borchers
  \cite{Bo1}. It ensures that the inner product is well-defined on
  test functions involving $\delta(\vec x)$.}
$$
f^X_{\underline b}(x)=f^X_{\underline b;\underline
  a}(t)\nabla_{\underline a}\delta(x)\equiv f^X_{b_1\dots
  b_h;a_1\dots a_k}(t)\cdot\nabla_{a_1}\dots\nabla_{a_k}\delta(\vec x)\quad (X=E,B)
  $$
(summation over $\underline a = a_i\dots a_k$ understood),
$k=0,1,2,\dots$. We call $T_k$ the subspace of such functions with a
fixed number $k$ of spatial derivatives and $T$ the union of all the $T_k$. 

The space of the test functions $f_{\underline b}^X$, modulo the kernel
of the inner product, defines the one-particle Hilbert space $\H$ of the field strength $F$, and by Haag duality $K(O_I)=\overline{\{F(f)\Omega: f\in T, \,\mathrm{Supp}\,f\subset I\}}\subset H(O_I)$ cf.\ \cite{Bo1}. Furthermore, by conformal covariance, the modular group of the double cone subspace $H(O_I)$ implements a one-parameter group of conformal transformations fixing the time axis (see \cite{H}) and any $T_k$ (the whole  representation $\wt U|_\Mob$ fixes $T_k$, cf.\ \eqref{eq:ray}).  Now one can see that $T+iT$ is cyclic in $\H$ since the inner product in the 
Hilbert space $\H$ can be decomposed as $(f,g)=\int p_0dp_0\,\int_{p_0\cdot S_2}(f,g)_pd\sigma$ where $p_0\cdot S_2$ is the sphere of radius $p_0$, $d\sigma$ is the $\SO(3)$-invariant measure on $p_0\cdot S_2$
and $\big(\cdot,\cdot\big)_{p}$ 
is a quadratic form involving $2h$ factors of $p_\mu=(p_0,\vec
p)=p_0(1,\vec n_\sig)$.\footnote{By the Stone theorem
  polynomials are dense in the continuous functions on the sphere. Then (vector) continuous functions are dense in the $L^2$-space w.r.t.\ the inner product $\int_{p_0\cdot S_2}(f,g)_{p}d\sigma$.} Then, by Lemma \ref{inc} we have that $K(O_I)=H(O_I)$.

Because of the symmetry of the tensors $E$ and $B$, it suffices to take
$f^X_{\underline b;\underline a}$ to be symmetric and traceless in the
$b$-indices; and because of the wave equation, it suffices to take it
also symmetric and traceless in the $a$-indices. Thus, the test functions
carry (twice) the representation $D^s\otimes D^k$.

The one-particle Hilbert space is defined by taking the quotient by
the null space, which is the kernel of the two-point function. Thus,
we may identify test functions according to (\ref{null}). In particular, every
test function in $T_k$ with coefficients $f^X_{\underline b;\underline
a}$ involving a factor $\delta_{b_ia_j}$ is zero in
the one-particle space; and every test function in $T_k$ with coefficients
anti-symmetric in a pair $b_i,a_j$ is identified with (the time
derivative of) a test function in $T_{k-1}$. Therefore, 
the one-particle Hilbert space for the restricted
fields is spanned by the spaces $\widetilde T_k$ ($k=0,1,2\dots$) with elements 
$$(f^E_{c_1\dots c_{h+k}},f^B_{c_1\dots c_{h+k}})\in \widetilde T_k$$ 
where $f^X_{\underline c}$ ($X=E,B$) are completely symmetric and
traceless, hence carrying (twice) the representation $D^{h+k}$ of
$SO(3)$. All other subrepresentations of $D^h\otimes D^k$ belong to
the null space. In particular $\widetilde T_k$ are mutually orthogonal.

We write the two-point function for $f=(f^E_{\underline
  c},f_{\underline c}^B)\in \widetilde T_k$ as 
\begin{equation}\label{eq:sp}
\big(f,f\big)_k = \int \frac{p_0^2\,dp_0}{p_0}\int d\sigma
\Big(\wh{f}(p_0,\vec p),\wh{f}(p_0,\vec
p)\Big)_{p}
\end{equation}
where $\wh{f}(p_0,\vec p)=(\wh{f^E}_{c_1\dots c_{h+k}}(p_0)p_{c_{h+1}}\dots
p_{c_{h+k}},\wh{f^B}_{c_1\dots c_{h+k}}(p_0)p_{c_{h+1}}\dots p_{c_{h+k}})$
are homogeneous polynomials of degree $k$ in $\vec p$.
Extracting powers of 
$\vert \vec p\vert=p^0$, this becomes 
$$\big(f,f\big)_k = \int p_0^{1+2h+2k}\,dp_0\int d\sigma
\Big(\wh{f}(p_0,\vec n_\sig),\wh{f}(p_0,\vec n_\sig)\Big)_{(1,\vec n_\sig)}.$$
The integration $d\sigma$ yields the inner
product for $D^{h+k}\oplus D^{h+k}$, while the M\"obius
transformations are characterized by the dependence on $p^0$.

The concluding argument is the same as in \cite{BDL}: The time translations
and dilations trivially restrict to the time axis by 
$$Pf(t)= i\partial_t f(t), \quad
Df(t)= i(t\partial_t-(h+k)) f(t).$$
The conformal inversion $I$ acts geometrically on test functions by $(t,\vec
x)\mapsto(-t,\vec x)/(t^2-\vec x^2)$, hence also its restricted action on
the time axis is geometric by $t\mapsto -1/t$. 
Because it commutes with $\so3$, it preserves the spaces $\widetilde T_k$ and
must act on it as 
\begin{equation}\label{eq:ray}
(If)(t)=G(t)f(I(t))\end{equation}
where $G(t)$ is a $2\times 2$-matrix in the commutant of
$D^{h+k}\oplus D^{h+k}$, possibly mixing the electric and magnetic
components. Now the argument of the Lemma \ref{l:geo} applies, and we conclude
that $G(t)=t^{2(h+k)}\eins_2$, and the subgroup $\Mob\times\so3$ of
$\C$ acts on $\widetilde T_k$ as $U^{(h+k+1)}\otimes (D^{h+k}\oplus D^{h+k})$.

We have proved the following theorem.
\begin{theorem}\label{thm:dec}
Let $U_{h}$ be the irreducible helicity-$h$ representation of the Poincar\'e group. Let $U=U_h\oplus U_{-h}$, and $\widetilde U$ its extension to the conformal group $\C$, then
\begin{equation}\label{eq:Udec}\widetilde  U|_{\Mob\times\SO(3)}=\bigoplus\nolimits_{k=0}^\infty
U^{(h+k+1)}\otimes (D^{h+k}\oplus D^{h+k}).
\end{equation}
\end{theorem}
\begin{corollary}\label{cor:dech} 
Let $U_h$ be the irreducible helicity-$h$ representation of the Poincar\'e group and $\widetilde U_h$ its extension to the conformal group $\C$, then
\begin{equation}\label{eq:dedec}\widetilde U_h|_{\Mob\times\SO(3)}=  \bigoplus\nolimits_{k=0}^\infty
U^{(h+k+1)}\otimes D^{h+k}.\end{equation}
\end{corollary}
\begin{proof}
The PCT symmetry respects the  $\Mob\times\SO(3)$ decomposition. Its 
anti-unitary implementation $J$ intertwines $U_h$, $U_{-h}$   and their restrictions to 
 $\Mob\times\SO(3)$. 
Irreducible unitary sub-representations in $\wt U|_{\Mob\times \SO(3)}$ of  $\Mob\times\SO(3)$  are tensor products of the form $U^{j+1}\otimes \D^j$ that anti-unitarily extend to $\Mob_2\times\SO(3)$. In particular, $\wt U_h|_{\Mob\times \SO(3)}$ and $\wt U_{-h}|_{\Mob\times\SO(3)}$ are unitarily equivalent, and  by the decomposition in Theorem \ref{thm:dec} we get the claim.
\end{proof}

\section{Trace class and split property for finite helicity fields}
\label{s:split}
\setcounter{equation}{0}
\begin{definition}\label{def:split}(Split Property) \cite{lodo}.
Let $(\N\subset \M,\Omega)$ be a \textit{standard inclusion} of von
Neumann algebras, i.e., $\Omega$ is a cyclic and separating vector
for $\N$, $\M$ and $\N'\cap \M$.

A standard inclusion $(\N\subset \M,\Omega)$ is \textit{split} if there exists a type I factor $\B$ such that $\N\subset \B \subset \M$.

A Poincar\'e covariant net $(\A,U,\Omega)$ satisfies the \textit{split property} if  the von Neumann algebra inclusion $(\A(O_1)\subset\A(O_2),\Omega)$ is split, for every compact inclusion of bounded causally closed regions $O_1\Subset O_2$.
\end{definition}
The following result relates the trace class property of the partition function in the first and second quantization nets.
\begin{lemma}\cite{BDL,L}\label{lem:trace}
Let $A\in\B(\H)$ be a selfadjoint operator s.t. $0\leq A<1$, then $\Tr\Gamma (A)<\infty$ iff $\Tr A<\infty$, where $\Gamma$ is the second quantization functor.
\end{lemma}
The next proposition relates the trace class and the split
properties of conformal nets on the circle.
\begin{proposition}\cite{BDL}\label{prop:BDLsplit}
Let $\A$ be a von Neumann algebra net on the circle satisfying the trace class condition $$\Tr\,e^{-\beta L_0}<\infty\quad \text{ for every }\beta>0,$$ then every inclusion $\A(I)\subset\A(\wt I)$, with $I\Subset \wt I$, is a split inclusion.
\end{proposition}
The results of the previous section allow us to conclude
\begin{proposition}  \label{prop:trh}
Let $U_h$ be a finite helicity representation of
  $\Poi$ and $\widetilde U_h$ its extension to the conformal group
  $\C$. Consider the restriction $\widetilde U_h|_\Mob$ and let $L_0$
  be the conformal Hamiltonian, i.e., the generator of the rotation 
subgroup of $\Mob$. Then $e^{-\beta L_0}$ is a  trace class operator.
\end{proposition}
\begin{proof} By Corollary \ref{cor:dech} any representation of highest weight $n+1$ in the decomposition of $U_h|_{\Mob\times\SO(3)}$ 
  appears with multiplicity  equal to the dimension of $\D^{n}$ when
  $n\geq s$: 
$$\widetilde U_h|_{\Mob\times \SO(3)}\simeq\bigoplus\nolimits_{k=0}^\infty U^{(h+k+1)}\otimes \D^{h+k}.$$ 
Furthermore, the trace of $L_0$ in $U^{(h+k+1)} $ is equal to $\frac{e^{-(h+k+1)\beta}}{1-e^{-\beta}}$. We conclude that
$$\Tr(e^{-\beta L_0}) =
\sum\nolimits_{n=h}^\infty(2n+1)\frac{e^{-(n+1)\beta}}{1-e^{-\beta}},$$
which converges for all $\beta>0$ as before.
\end{proof}
\begin{proposition}\label{prop:OI}
Let $\A_h$ be the helicity-$h$ free net of von Neumann algebras (whose
one-particle space carries the representation $U_h\oplus U_{-h}$ if $h>0$) and
$I\mapsto\A_h(I)\doteq\A_h(O_I)$ be its restriction to the time axis. 
Then $\A_h(I)\subset\A_h(\wt I)$ is a split inclusion when
$I\Subset\wt I$. 
\end{proposition}
\begin{proof}
 The net $\A_h$ is the second quantization of the BGL net $H_h$ of standard subspaces associated with $U=U_h\oplus U_{-h}$. 
By Lemma \ref{lem:trace} and Proposition \ref{prop:trh}, we have that  $\Tr \Gamma (e^{-\beta L_0})<\infty$, thus the net 
$$I\mapsto\A_h(O_I)$$ satisfies the split property, by Proposition \ref{prop:BDLsplit}.
\end{proof}
\begin{theorem}
The free finite helicity fields satisfy the split property.
\end{theorem}
\begin{proof}
For inclusion of algebras related to double cones on the time axis, we conclude by Proposition \ref{prop:OI}.

For a general inclusion of double cones $O\Subset\wt O$, choose a
Poincar\'e transformation $g$ such that $g(\wt O)=O_{\wt I}$ is
a double cone on the time axis. Then there is an inclusion $O_I\Subset
O_{\wt I}$ of another double cone on the time axis such that
$g(O)\subset O_I$. Then $\A_h(g(O))\subset \A_h(O_{\wt I})=\A_h(g(\wt
O))$ 
is split because $\A_h(g(O))\subset \A_h(O_I)$, and hence $\A_h(O)\subset \A_h(\wt
O)$ is split by covariance.
\end{proof}

As a corollary of Proposition \ref{prop:trh} we also have the
$L^2$-nuclearity property, which is stronger than the split property.  
\begin{corollary} ($L^2$-nuclearity)
Let $\A_h$ be the helicity-$h$ free net of von Neumann algebras
and $I\mapsto\A_h(I)\doteq\A_h(O_I)$ be its restriction to the time
axis. Then for $I\Subset \wt I$ the operator $\Delta_{\A(\wt I),\Omega}^{1/4}\Delta_{\A(I),\Omega}^{-1/4}$ is trace class.
\end{corollary}
The proof of the corollary is analogous to the one given in \cite{BDL}.

\section{Outlook: Towards a new construction of finite helicity fields}
\label{s:defo}
\setcounter{equation}{0}
Disjoint unitary representations of a given locally compact group $G$
can have unitary equivalent restrictions to subgroups. This fact can be
used to reconstruct inequivalent representations of $G$, by perturbing
generators in the complement of a subgroup $H\subset G$.  In
\cite{GLW} the authors proved that inequivalent highest weight
representations of the $\Mob$ group have unitary equivalent
restrictions to the translation-dilation subgroup, cf.\ Sect.\
\ref{s:mob}. In particular one can recover the full $\Mob$
representation $U^{(n)}$ of lowest weight $n$ by perturbing the conformal
inversion operator of the representation $U^{(1)}$ 
  of lowest weight $1$.
On the other hand, the covariance of associated nets is  not preserved
in this perturbation procedure.  For instance, one can see that
$U^{(n)}$ acts covariantly only on a subnet of the $U(1)$-current (which
anyway coincides with the $U(1)$-current on half-lines) \cite{GLW}.

In this paper we established the split property for free finite
helicity fields. The fundamental step is the factor decomposition of
the restriction of $\wt U_h$, the extension of the representation
$U_h$ of helicity $h$ to the conformal group $\C$, to the subgroup $\Mob\times \SO(3)$. 
The rotation group $\SO(3)$ is a type $I$ group, hence irreducible
representations of $\Mob\times\SO(3)$ have to be tensor products
$U^{(n)}\otimes \D^s$,  where $U^{(n)}$ is the lowest weight-$n$
representation of $\Mob$ and $\D^s$ is the spin-$s$ representation of
$\SO(3)$. By inspection of the decomposition of $\wt U_h|_{\Mob\times \SO(3)}$ in Corollary \ref{cor:dech}, we observe that $\wt U_{h_1}|_{\Mob\times \SO(3)}$ is a sub-representation of $\wt U_{h_2}|_{\Mob\times \SO(3)}$ when $h_1-h_2\in\ZZ$ and $h_1\geq h_2$. 

One can think of a perturbation argument. Consider the projection
$P_h$ on the subspace supporting $\bigoplus_{k=0}^{h-1}(U^{(k+1)}\otimes
\D^k)$ and cut $\wt U_0|_{\Mob\times\SO(3)}$ along the
complementary space $1-P_h$. By Corollary \ref{cor:dech} the
representation $U_0|_{\Mob\times \SO(3)}(1-P_h)$ extends to a
representation of helicity $h$ by redefining the spatial translations,
suitably perturbing the translation generators in the scalar
representation on $(1-P_h)\H_0$.  Namely, the spatial translations
together with the time translations and the conformal inversion,
contained in $\Mob$, generate the conformal group.
This can be further seen by looking at the proof of Proposition
\ref{prop:decomp}, where we disintegrate the spectrum in
rotation-translation invariant fibers, and \eqref{eq:einc} shows that 
$W_{h,p_0}|_{\SO(3)}\leq W_{k,p_0}|_{\SO(3)}$, for $k\leq h$. Thus one can
address the perturbation argument already at the  level of the Euclidean subgroup, cf.\ Appendix \ref{app:eucl}.

Let us comment on inclusions of nets of standard subspace on the time axis.
Firstly, the BGL-net associated with the $\Poi$-representation $U_0$
extends to a conformal net, and  the Bisognano-Wichmann property for
boosts and dilations is a consequence of conformal covariance, cf.\
\cite{BGL93}. Then, we note that the projection $P_{h+1}-P_{h}$
commutes with $U_0|_{\Mob}$ (and with $U_0|_{\Mob\times \SO(3)}$). In
particular, the net on the time axis $I\mapsto H_0(O_I)$ decomposes as
the direct sum of $\Mob$-covariant nets of subspaces
$$I\mapsto H_0(O_I)=\bigoplus_{h=0}^{\infty}\left(P_{h+1}-P_{h}\right) H_0(O_I)$$ 
according to \eqref{eq:dedec}. Indeed, by the Bisognano-Wichmann property the modular groups of the  interval subspaces implement interval dilations,  the interval modular conjugations implement the PCT symmetry, and it is easy to see that the Tomita operators of the interval subspaces commute with $P_{h+1}-P_h$.

Once we identify the representations $$(\wt U_h\oplus \wt U_{-h})|_{\Mob\times \SO(3)} = \left((1-P_h)\oplus(1-P_h)\right)( \wt U_0\oplus \wt U_0)|_{\Mob\times \SO(3)}\left((1-P_h)\oplus(1-P_h)\right),$$ by Proposition \ref{prop:uuu}, we can also identify $H_h(I)$ as a subnet of $H_0(I)\oplus H_0(I)$:
take two copies of the massless scalar one-particle net $(U_0\oplus
U_0, H_0\oplus H_0)$ and consider the net on the time axis   $I\mapsto
H_0(O_I)\oplus H_0(O_I)$; then  consider the $\Mob\times\SO(3)$
invariant projections $(1-P_h)\oplus (1-P_h)$ and the new net on the
time axis
$$I\mapsto H_h(I)\dot=\left((1-P_h)\oplus (1-P_h) \right)\left(H_0(O_I)\oplus H_0(O_I)\right),$$
undergoing the $\Mob$ (and $\Mob\times\SO(3)$)-action through $(\wt
U_h\oplus \wt U_{-h})|_{\Mob\times \SO(3)} $. The projection
$1-P_h$ does not commute with the $U_0$-translations since $U_0$ is
irreducible, and on the subspace $\left((1-P_h)\oplus
  (1-P_h)\right)(\H_0\oplus\H_0)$ one has to define new translations
to obtain the $U_h\oplus U_{-h}$ representation of the Poincar\'e
group (the group generated by $\Mob\times SO(3)$ and space
translations contains the Poincar\'e group). Afterwards, one can
define by covariance double cone subspaces and the helicity-$h$ free
net of standard subspaces, namely $$H_h(O)\dot=\left(U_h(g)\oplus
  U_{-h}(g)\right)H_h(O_I)$$ for a general double cone $O=gO_I$.  It
remains an interesting open problem to explicitly provide or
characterize the necessary perturbation of the $U_0$-translations in order to obtain the $U_h$-translations on the proper subspace. 

This further suggests another way of constructing finite helicity free
nets. One can start with the representation of $\Mob \times \SO(3)$ in
the right-hand side of \eqref{eq:dedec}. It extends to the representation of the Poincar\'e group of helicity $h$ or $-h$. Consider two copies of such a $\Mob \times \SO(3)$-representation, and
the associated one-particle net on the line can be identified with the
time axis theory of the helicity-$h$ free net. Then there is a proper choice of the translation generators and the PCT operator which allows to construct the free net on the full Minkowski space by covariance.

\appendix

\section{Appendix: Restriction of finite helicity representations to the Euclidean subgroup}\label{app:eucl}
\setcounter{equation}{0}
We comment on the restriction of finite helicity representations to the Euclidean group.

\begin{definition}\cite{Mac} Let $G$ be a separable locally compact group.

Closed subgroups $G_1$ and $G_2$ of
$G$ are said to be {\it regularly related} if there exists a sequence $E_0, E_1, E_2, \ldots$ of measurable subsets of $G$ each of which is a union of $G_1:G_2$ double cosets such that $E_0$ has Haar measure zero and each double coset not in $E_0$ is the intersection of the $E_j$ which contain it.

Because of the correspondence between orbits of $G/G_1$ under
$G_2$ and $G_1:G_2$ double cosets, $G_1$ and $G_2$ are regularly related if and only if the orbits outside of a certain set of measure zero form the equivalence
classes of a measurable equivalence relation. 

Consider the map $s: G\rightarrow G_1\backslash G\slash G_2$  carrying each element of $G$ into its double coset. Then equip $G_1\backslash G\slash G_2$ with the quotient topology given by $s$ and consider a finite measure $\mu$ on $G$ which is in the same measure class of the Haar measure. It is possible to define $\ol \mu$ on the  Borel sets of $G_1\backslash G\slash G_2$ by $\ol\mu(E)=\mu(s^{-1}(E))$. We shall call $\ol \mu$ an \textit{admissible measure} in $G_1\backslash G\slash G_2$. The definition is well posed since any two of such measures have the same null measure sets.

\end{definition}

We recall two well-known theorems. 
\begin{theorem}[Mackey's subgroup Theorem]\label{thm:sgr}\cite{Mac2}.
Let $G_1,\, G_2$ regularly related in $G$. Let $\pi\in \mathrm{Rep}(G_1)$. For each $x\in G$ consider $G_x=G_2\cap(x^{-1}G_1x)$ and set
$$V_x=\mathrm{Ind}_{G_x\uparrow G_2}(\pi\circ\ad\, x).$$
Then $V_x$ is determined to within equivalence by the double coset
$\ol x$ to which $x$ belongs. If $\nu$ is an
admissible measure on $G_1\backslash G\slash G_2,$ then
$$\left(\mathrm{Ind}_{G_1\uparrow G}\,\pi\right)|_{G_2}\simeq\int_{G_1\backslash G\slash G_2}V_{\ol x}\,d\nu(\ol x).$$
\end{theorem}
If $G$ is a compact group, let $\pi$ and $\rho$ be two unitary representations of $G$, we shall denote with $\C(\pi,\rho)$ the space of intertwining operators of the representations $\pi$ and $\rho$ and with $\mathrm{mult}(\pi,\rho)$ the multiplicity (of the unitary class) of $\pi$ in $\rho$.
\begin{theorem}[Frobenius Reciprocity theorem]\cite{foll}\label{rec}.
Let $G$ a compact group, $H$ a closed subgroup, $\pi$ a unitary representation of the group $G$, and $\sigma$ an irreducible unitary representation of $H$. Then, 
$$\C(\pi,\ind_{H\uparrow G}(\sigma))\simeq\C(\pi |_H,\sigma)\quad\text{and}\quad \mathrm{mult}\,(\pi,\ind_{H\uparrow G}(\sigma))=\mathrm{mult}\,(\pi |_H,\sigma).$$
\end{theorem}
In this section we shall indicate with $\chi$ a one-dimensional representation (a character) of an abelian group. Let $E(n)=\SO(n)\ltimes\RR^n$ be the inhomogeneous symmetry group of $n$-dimensional Euclidean space. The universal covering is the semi-direct product $\widetilde E(n)=\widetilde{\SO}(n)\ltimes\RR^{n}$. Representations are obtained by induction. Consider a character $\chi_q$ in the dual of the translation group and its orbit $\sigma_q$ through the dual action of $\widetilde E(n)$. We shall call $\Stab_q$ and $\overline\Stab_q=\Stab_q\ltimes\RR^{n}$ the stabilizers of $\chi_q$ through the $\widetilde{\SO}(n)$ and $\widetilde E(n)$ actions, respectively. Note that the dual action of the translations is trivial on $\chi_q$. When there is no ambiguity we will write $q$ instead of $\chi_q$.

There are two main families of irreducible representations (cf.\ e.g.\ \cite{var,foll}):
\begin{itemize}\itemsep0mm
\item $U=\ind_{\overline\Stab_0\uparrow \widetilde E(n)}\chi_0 V=V$ is induced from a product of the trivial character $\chi_0$ of $\RR^n$ and an irreducible representation $V$ of $\Stab_0=\widetilde \SO(n).$ Thus $U$ is the irreducible representation $V$ of $\widetilde \SO(n)$ lifted to $\widetilde E(n)$, trivial on translations;

\item  $U=\ind_{\overline\Stab_q\uparrow \widetilde E(n)}\chi_q V'$ is induced from a product of a nontrivial character $\chi_q$ of $\RR^n$ and an irreducible representation $V'$ of $\Stab_q.$ In such a case orbits are spheres of radius $r=|q|$ and up to unitary equivalence it is possible to choose $q=(\textbf{0},r)$ where $\textbf{0}$ is the null vector in $\RR^{n-1}$. 
\end{itemize}
In the  three-dimensional Euclidean case, if $ q=(0,0,r)$ with $r>0$ then $\Stab_q=U(1)$, double covering of $\SO(2)$. Induced representations are of the form
$$W_{h,r}=\Ind_{\overline\Stab_q\uparrow \widetilde E(3)}\chi_q\chi_h,\qquad h\in\frac12\,\ZZ,$$
where $\chi_h$ is the $2h$-character $U(1)$-representation and $q$ defines a character of $\RR^3$ of length $r$.
The induced representation acts on the Hilbert space $L^2(S_r,dp\, \delta(\textbf{p}^2-r^2))$ where $S_r$ is the sphere with center in the origin and radius $r$.

\begin{proposition}\label{prop:decomp}
Let $U_h$ be a massless helicity-$h$ representation. Consider the restriction of $U_h$ to $T\times E(3)$, where $T$ is the time-translation group, then
\begin{equation}\label{eq:qqq}U_h|_{T\times E(3)}=\int_{\RR^+}dp_0\left( \chi_{p_0} \otimes\,W_{h,p_0}\right).\end{equation} Furthermore, \begin{equation}\label{eq:einc}W_{h,p_0}|_{\SO(3)}=\bigoplus_{l=|h|}^\infty\D^l\end{equation} and \begin{equation}
\label{eq:dec}
 U_h|_{T\times \SO(3)}=\bigoplus_{l=|h|}^\infty\int_{\RR^+}dp_0\left( \chi_{p_0} \otimes\,\D^l\right)\end{equation}
\end{proposition}
\begin{proof}We prove the proposition in the bosonic case, namely $h\in\ZZ$. The proof is analogous in the Fermionic case.

Let $q=(1,0,0,1)$, with the definitions in Sect. 2.2, the helicity-$h$ representation is 
$$U_h=\ind_{\overline{\Stab_q}\uparrow\Poi} \chi_qV_h$$ 
where $\overline\Stab_q=E(2)\ltimes\RR^4\subset \Poi$. 
When we restrict $U_h$ to $T\times E(3)$,  we get
\begin{equation}\label{eq:trtr}U_h|_{T\times E(3)}=\int_{\RR^+}d \mu(p_0)\,\chi_{p_0}\, U_{p_0}
\end{equation} 
where $U_{p_0}$ are representations of $E(3)$ of radius $p_0$, and $\mu$ is a Borel measure on $\RR^+$. 
This follows since $$(U_{h}(a,A)\phi)(p)=e^{ia\cdot p}V_h(B_p^{-1}AB_{A^{-1}p})\phi(A^{-1}p),\;\; (a,A)\in\Poi, \phi\in L^2(\partial V_+,\theta(p_0)\delta(p^2)d^4p),$$
and we can choose $B_p^{-1}=\Lambda_{3}(-\ln p_0)R_{p}$, where
$p\mapsto R_{p}$ is a Borel map from the $\RR^3$-sphere $S_{p_0}$ of
radius $p_0$ (we are considering the set
$(p_0,S_{p_0})\subset\RR^{1+3}$) to $\SO(3)$ such that,  for any  $p$,
$R_p\,p= q_{p_0}\dot=(p_0,0,0,p_0)$ and  $\Lambda_3$ is the
$x_0$-$x_3$ boost s.t.\ $\Lambda_3(-\ln p_0)q_{p_0}=q=(1,0,0,1)$ (cf. \cite{var}). Thus, with $U_{p_0}=e^{ia\cdot p}V_h(R_p^{-1}AR_{A^{-1}p})\psi(A^{-1}p)$ 
where $(a,A)\in E(3)$ and $\psi\in L^2((p_0,S_{p_0}),\delta(p^2)\delta(p_0)d^4p)$, the direct integral 
representation of $T\times E(3)$ in the right hand side of \eqref{eq:trtr}  extends to the representation 
of the Poincar\'e group $U_h$.

Now, with $\delta_t:p\mapsto e^{\lambda}p$ the dilation group,  by dilation covariance of $U_h$
$$U|_{T\times E(3)}= \int_{\RR^+}d\mu(p_0)\,
U_{p_0}\chi_{p_0}\simeq\int_{\RR^+}d\mu(p_0)\,
U_{\delta_{-t}p_0}\chi_{\delta_{-t}p_0}=\int_{\RR^+}d\mu_t(p_0)\,
U_{p_0}\chi_{p_0},$$ 
thus $U_{p_0}\chi_{p_0}{\simeq} U_{\delta_{-t} p_0'}\chi_{\delta_{-t}{p_0'}}$ for $\delta_{-t}
{p_0}'=p_0$ with  $\lambda\in\RR$ and $\mu$ is equivalent to
$\mu_t(p_0)=\mu(e^{t}p_0)$, hence $\mu$ is equivalent to the Lebesgue
measure.  
In particular, $U_{p_0}$ is irreducible for almost every $p_0\in\RR^+$ because the $U_h$-translation algebra is multiplicity free. 

Since the stabilizer of $q_{p_0}$ under the $(T\times E(3))$-action is contained in $E(2)\ltimes\RR^4\subset\Poi$ (the stabilizer of $q_{p_0}$ under the Poincar\'e action) and $V_h$ is trivial on $E(2)$-translations, then one can see that for almost every $p_0\in\RR^+$, $U_{p_0}=\Ind_{\SO(2)\ltimes \RR^3\uparrow E(3)}\chi_h\chi_{q_{p_0}}$ and we get
$$U_h|_{E(3)}=\int_{\RR^+}d\mu(p_0)\,  W_{h,p_0}\chi_{p_0}.$$
Now, we apply Theorem \ref{thm:sgr} to $W_{h,p_0}$ with $G=E(3)$,
$G_1=\SO(2)\ltimes\RR^3$, $G_2=\SO(3)$ (note that $G_1\backslash
G\slash G_2={\bf1}$). By Theorem \ref{rec}, we get the second statement, i.e.,
$$W_{h,p_0}|_{\SO(3)}=\bigoplus_{l=|h|}^\infty\D^l.$$ \eqref{eq:einc} does not depend on the radius $p_0$, thus we conclude \eqref{eq:dec}.
\end{proof}

\bigskip 

\textbf{Note added.} The suggestions discussed in the
  Outlook are confirmed in \cite{MR}.

\bigskip 

\textbf{Acknowledgement.} V.M.\ thanks Gerardo Morsella and  Massimo
Bianchi for valuable discussions. R.L.\ and V.M.\ acknowledge the MIUR
Excellence Department Project awarded to the Department of
Mathematics, University of Rome Tor Vergata, CUP E83C18000100006.


\begin{thebibliography}{00}\itemsep-0.9mm
\bibitem{A} 
{\sc H. Araki}, 
\textit{A lattice of von Neumann algebras associated with the quantum theory of a free Bose field}, J. Math.\ Phys.\ \textbf{4} (1963) 1343--1362.
%
\bibitem{BW} 
{\sc J.J. Bisognano, E.H. Wichmann}, 
\textit{On the duality condition for quantum fields}, J. Math.\ Phys.\ \textbf{17} (1976) 303--321.
%
\bibitem{Bo1}{\sc H.-J. Borchers}, \textit{Field operators as $C^\infty$
  functions in spacelike directions}, Nuovo Cim.\ {\bf 33} (1964) 1600--1613.
%
\bibitem{Bo2}{\sc H.-J. Borchers}, \textit{The CPT-theorem in two-dimensional
  theories of local observables}, Commun.\ Math.\ Phys.\ {\bf 143}
  (1992) 315--332.
%
\bibitem{BGL} {\sc R. Brunetti, D. Guido, R. Longo}, 
{\it Modular localization and Wigner particles}, 
Rev.\ Math.\ Phys.\ {\bf 14} (2002) 759--786. 
%
\bibitem{BGL93}{\sc R. Brunetti, D. Guido, R. Longo}, 
{\it Modular structure and duality in conformal quantum field theory}, 
Commun.\ Math.\ Phys.\ \textbf{156} (1993) 201--219. 
%
\bibitem{B74}
{\sc D. Buchholz}, \textit{Product states for local algebras}, 
  Commun.\ Math.\ Phys.\ {\bf 36} (1974) 287--304.
%
\bibitem{BDLI}
{\sc D. Buchholz, C. D'Antoni, and R. Longo},
\textit{Nuclear maps and modular structures. {I}.\ {G}eneral properties},
J. Funct.\ An.\ {\bf 88} (1990) 233--250.
%
\bibitem{BDL}{\sc D. Buchholz, C. D'Antoni, R. Longo}, 
    \textit{Nuclearity and thermal states in conformal field theory},
    Commun.\ Math.\ Phys.\ {\bf 270} (2007) 267--293.
%
\bibitem{BW86}
{\sc D. Buchholz, E. H. Wichmann},
\newblock \textit{Causal independence and the energy-level density of states in local
  quantum field theory}, Commun.\ Math.\ Phys.\ {\bf 106} (1986) 321--344.
%
\bibitem{lodo} {\sc S. Doplicher, R. Longo}, \textit{Standard and split
    inclusions of von Neumann algebras}, Invent.\ Math.\ {\bf 75}
  (1984) 493--536.

\bibitem{DL} {\sc S. Doplicher, R. Longo}, {\it Local aspects of superselection
rules. II}, Commun.\ Math.\ Phys.\ 88 (1983) 399--409.
%
\bibitem{foll} {\sc G.B. Folland}, A course in abstract harmonic analysis. Studies in Advanced Mathematics. CRC Press, Boca Raton, FL, 1995.
%
\bibitem{GL95}{\sc D. Guido, R. Longo}, \textit{An algebraic Spin and
    Statistics Theorem}, Commun.\ Math.\ Phys.\ {\bf 172} (1995) 517--533.
%
\bibitem{GLW} {\sc D. Guido, R. Longo, H.-W. Wiesbrock}, \textit{Extensions
    of conformal nets and superselection structures}, 
    Commun.\ Math.\ Phys.\ \textbf{192} (1998)  217--244.
%
\bibitem{HO}{\sc D. Harlow, H. Ooguri}, \textit{Symmetries in quantum field theory and quantum gravity}, arXiv:1810.05338 (2018).
%
\bibitem {H}  {\sc P.D. Hislop}, \textit{Conformal covariance,
    modular structure, and duality for local algebras in free massless
    quantum field theories}, 
Ann.\ Phys.\ {\bf 185} (1988) 193--230. 
%
\bibitem{HL} {\sc P.D. Hislop, R. Longo}, 
{\it Modular structure of the local algebras associated with the free massless scalar field theory}, 
Commun.\ Math.\ Phys.\ {\bf 84} (1982) 71--85. 
%
\bibitem{HS}{\sc S. Hollands, K. Sanders}, 
\textit{Entanglement measures and their properties in quantum field theory}, arXiv:1702.04924 (2017).
%
\bibitem{LRT} {\sc P. Leyland, J.E. Roberts, D. Testard}, 
\textit{Duality for quantum free fields}, 
unpublished manuscript, Marseille 1978. 
%
\bibitem{L}{\sc R. Longo}, 
``Lectures on Conformal Nets'', preliminary lecture notes that are available at
http://www.mat.uniroma2.it/{$\sim$}longo/Lecture-Notes.html\ . 
%
\bibitem{LN}
{\sc R. Longo},  
\textit{Real Hilbert subspaces, modular theory, $\SL2$ and CFT}, in:
``Von Neumann algebras in Sibiu'', pp.\ 33--91, {\it Theta
  Ser.\ Adv.\ Math.}, {\bf 10}, Theta, Bucharest (2008). 
%
\bibitem{LMR}{\sc R. Longo, V. Morinelli, K.-H. Rehren}, \textit{Where infinite spin particles are localizable}, Commun.\ Math.\ Phys.\ \textbf{345} (2016) 587--614.
%
\bibitem{Mack} {\sc G. Mack}, \textit{All unitary ray representations of the
    conformal group {${\rm SU}(2, 2)$} with positive energy}, 
    Commun.\ Math.\ Phys.\ {\bf 55} (1977) 1--28.
%
\bibitem{Mac} {\sc G.W. Mackey}, \textit{Induced representations of locally
  compact groups. I.}, Ann.\ Math.\ {\bf 55} (1952) 101--139.  
%
\bibitem{Mac2} {\sc G.W. Mackey}, The theory of unitary group
  representations. Chicago Lectures in
  Mathematics. University of Chicago Press, Chicago, Ill., 1955; London, 1976. 
%
\bibitem{Mo} {\sc V. Morinelli}, {\it The Bisognano-Wichmann property
    on nets of standard subspaces, some sufficient conditions}, Ann.\
  H. Poinc.\ {\bf 19} (2018) 937--958.

\bibitem{MR} {\sc V. Morinelli, K.-H. Rehren}, {\it Spacelike
    deformations: Higher-spin fields from scalar fields},
  arXiv:1905:08714 (2018).
%
\bibitem{MTW} {\sc V. Morinelli, Y. Tanimoto, M. Weiner},
  \textit{Conformal covariance and the split property},
  Commun.\ Math.\ Phys.\ {\bf 357} (2018) 379--406.
%
\bibitem{NO} {\sc K.-H. Neeb, G. Olafsson}, \textit{Antiunitary representations and modular theory}, arXiv:1704.01336 (2017).
%
\bibitem{O}{\sc K. Osterwalder}, \textit{Duality for free Bose fields}, Commun.\ Math.\ Phys.\ {\bf 29} (1973) 1--14.
%
\bibitem{RV} {\sc M.A. Rieffel, A. Van Daele}, 
\textit{A bounded operator approach to Tomita-Takesaki theory}, 
Pacific J. Math.\ \textbf{69} (1977) 187--221.
%
\bibitem{var}{\sc V.S. Varadarajan}, Geometry of quantum theory,
  {Second Edition}, Springer-Verlag, New York (1985).
%
\bibitem{We} {\sc S. Weinberg}, The quantum theory of
  fields. Vol. I. Foundations. Cambridge University Press, Cambridge,
  2005. 
%
\bibitem{WW} {\sc S. Weinberg, E. Witten}, \textit{Limits on massless particles},
Phys.\ Lett.\ \textbf{B96} (1980) 59--62.
%
\bibitem{Wi}{\sc E.P. Wigner}, \textit{On unitary representations
  of the inhomogeneous Lorentz group}, Ann.\ Math.\ {\bf 40} (1939) 149--204.

\end{thebibliography}
\end{document}